\documentclass[pdflatex,sn-mathphys,Numbered]{sn-jnl}% Math and Physical Sciences Reference Style
\usepackage{faktor}
\usepackage{mathtools} %for \coloneqq
\usepackage{physics} %for \tr, \bra, \ket, \expval, \matrixel, etc.
\usepackage{comment}
\usepackage{graphicx}%
\usepackage{multirow}%
\usepackage{amsmath,amssymb,amsfonts}%
\usepackage{amsthm}%
\usepackage{mathrsfs}%
\usepackage[title]{appendix}%
\usepackage{xcolor}%
\usepackage{textcomp}%
\usepackage{manyfoot}%
\usepackage{booktabs}%
\usepackage{algorithm}%
\usepackage{algorithmicx}%
\usepackage{algpseudocode}%
\usepackage{listings}%
\newtheorem{theorem}{Theorem}
\numberwithin{theorem}{section}
\newtheorem{lemma}[theorem]{Lemma}
\newtheorem{Def}[theorem]{Definition}
\newtheorem{prop}[theorem]{Proposition}

\newtheorem{remark}[theorem]{Remark}

\numberwithin{equation}{section}
\newcommand{\up}{\uparrow}
\newcommand{\dn}{\downarrow}
\newcommand{\s}{\mathfrak{S}}

\newcommand{\G}{\mathcal{G}}
\newcommand{\I}{\mathrm{i}}
\newcommand{\mean}{\mathbf{E}}
\newcommand*\Bell{\ensuremath{\boldsymbol\ell}}

\hypersetup{
	colorlinks = true,
	citecolor = teal,
	linkcolor=red,
	urlcolor=blue}

\begin{document}

\title[ETH for Translation Invariant Spin Systems]{%
	Eigenstate Thermalisation Hypothesis for Translation Invariant Spin Systems}


\author*[1]{\fnm{Shoki} \sur{Sugimoto}}\email{sugimoto@cat.phys.s.u-tokyo.ac.jp}
\author[2]{\fnm{Joscha} \sur{Henheik}}\email{joscha.henheik@ist.ac.at}
\author[2]{\fnm{Volodymyr} \sur{Riabov}}\email{volodymyr.riabov@ist.ac.at}
\author[2]{\fnm{László} \sur{Erd\H{o}s}}\email{lerdos@ist.ac.at}

% \author{\small László Erd\H{o}s, Joscha Henheik, Volodymyr Riabov, and Shoki Sugimoto}

\affil[1]{\orgdiv{Department of Physics}, \orgname{the University of Tokyo}, \orgaddress{\street{7-3-1 Hongo}, \city{Bunkyo-ku}, \state{Tokyo}, \postcode{113-0033}, \country{Japan}}}
\affil[2]{\orgname{IST Austria}, \orgaddress{\street{Am Campus 1}, \city{3400 Klosterneuburg}, \country{Austria}}}

\abstract{
	    We prove the Eigenstate Thermalisation Hypothesis (ETH) for local observables in a typical translation invariant system of quantum spins with mean field interaction.
    This mathematically verifies the observation made in Ref.~\cite{santos2010localization} that ETH may hold for systems with additional translation symmetries for a naturally restricted class of observables.
    We also present numerical support for the same phenomenon
    for Hamiltonians with  local interaction.
}
\keywords{Eigenstate thermalisation, Microcanonical ensemble, Translational invariance, Quantum spin systems}
\pacs[AMS Subject Classification]{60B20, 82B44, 82D30}

\maketitle
\vspace{-10mm}
\emph{Date: \date{\today}}

%TODO:
%\begin{itemize}
%	\item Check/add/remove references (to make sure that there is no superfluous one, but maybe your bibtex
	%program already takes care of that)
	%\item comment $\Delta$ restrictions for mc ensemble after Theorem \ref{thm_mainTheoerem} (and maybe Theorem~\ref{thm_mainTheoerem_Ddimension_HighDimension} again)
	%\item write Remark on ETH with sup over local observables
	%\item Write proof of Lemma \ref{lemma_ETHwithinMomentumSector}
	%\item complete numerics section \ref{sec:numerics}
%\end{itemize}

%%%%%%%%%%%%%%%%%%%%%%%%%%%%%%%%%%%%%%%%%%%%%%%%%%
%%%%%%%%%%%%%%%%%%%%%%%%%%%%%%%%%%%%%%%%%%%%%%%%%%
\section{Introduction}
	Recent experiments have demonstrated thermalisation of isolated quantum systems under unitary time evolution~\cite{trotzky2012probing, langen2013local, clos2016time, kaufman2016quantum, neill2016ergodic, tang2018thermalization}.
	 In this context, thermalisation means that, after a long time evolution, observables attain their equilibrium (thermal) values determined by statistical mechanics.
	The primary mechanism behind this thermalisation of isolated quantum systems is an even stronger concept,
	the \textit{Eigenstate Thermalisation Hypothesis~(ETH)}~\cite{von2010proof, deutsch1991quantum, srednicki1994chaos}.
	Informally, the ETH asserts that (i) \textit{physical} observables $A$ take their thermal value on every eigenstate of a many-body quantum system
	%every eigenstate of a many-body quantum system is thermal \textit{per se} in terms of \textit{physical} observables $A$
	and (ii) off-diagonal elements of $A$ in the energy eigenbasis are vanishingly small.
	In particular,
	the ETH ensures the thermalisation of $A$ for any initial state with a macroscopically definite energy,
	given no massive degeneracy in the energy spectrum~\cite{d2016quantum, mori2018thermalization, deutsch2018eigenstate}.
	The ETH has numerically been verified for  individual models with several local or few-body
	observables~\cite{rigol2008thermalization, rigol2009breakdown, rigol2009quantum, biroli2010effect, santos2010localization, steinigeweg2013eigenstate, beugeling2014finite}.
	On the other hand, recent studies have revealed several classes of systems for which the ETH breaks down: examples include systems with an extensive number of local conserved quantities~\cite{rigol2007relaxation, cassidy2011generalized, hamazaki2016generalized}, many-body localisation~\cite{basko2006metal, nandkishore2015many}, and quantum many-body scars~\cite{turner2018weak, bull2019systematic}.
	%Therefore, a natural question is how \textit{generically} the ETH applies to realistic systems.

	As another approach to this question,
	it has been proven that the ETH holds true for \textit{any} deterministic observable for \textit{almost all} Hamiltonians $H$~\cite{von2010proof, reimann2015generalization, cipolloni2021eigenstate} sampled from a Wigner matrix ensemble
	which has no further unitary symmetry  (see also~\cite{cipolloni2023gaussian, adhikari2023eigenstate}
	for ETH for more general mean field ensembles). \normalcolor
	If the Hamiltonian has some unitary symmetry, the ETH clearly breaks down for conserved quantities related to those symmetries because we can find simultaneous eigenstates of the Hamiltonian and conserved quantities.
	However, Ref.~\cite{santos2010localization} observed an interesting phenomenon, namely that local quantities still satisfy the ETH even in a system with translational symmetry.
	Therefore, the question of how \textit{generically} and for what class of observables the ETH holds true in realistic situations has yet to be fully resolved.

	In this paper we mathematically rigorously prove an instance of the observation from~\cite{santos2010localization}.
	More precisely we show that, for the mean-field case of an ensemble with translational symmetry, the ETH typically holds for quantities whose support does not exceed half of the system size with the optimal speed of convergence.
	The ETH also typically holds for quantities whose support exceeds half the system size but with a slower convergence speed, while it typically breaks down for some observables whose support extends to the entire system.
	We complement our analytical results for the mean-field case with a numerical simulation for an ensemble of more realistic Hamiltonians with local interactions.

%%%%%%%%%%%%%%%%%%%%%%%%%%%%%%%%%%%%%%%%%%%%%%%%%%
%%%%%%%%%%%%%%%%%%%%%%%%%%%%%%%%%%%%%%%%%%%%%%%%%%
\section{Setup}
	We consider a one-dimensional periodic quantum spin system on the $L \in \mathbb{N}$ sites of the standard discrete torus
	$$\mathbb{T}_{L} \coloneqq \faktor{\mathbb{Z}}{L\mathbb{Z}}\,. $$
	On each vertex $j \in \mathbb{T}_L$, the one particle Hilbert space $\mathcal{H}_j$ is given by $\mathbb{C}^2$ and we denote its canonical basis by $\{ \ket{\uparrow}, \ket{\downarrow}\}$. The corresponding $L$-particle Hilbert space
	$$\mathcal{H} \coloneqq \bigotimes_{j=1}^{L} \mathbb{C}^{2}$$
	is simply given by the tensor product with dimension $2^{L}$.
	For simplicity, we restrict ourselves to the spin-$1/2$ case, but our results can straightforwardly be extended to general spin $s$ with one particle Hilbert space being $\mathbb{C}^{2s+1}$.

	Next, we introduce the ensemble of Hamiltonians, which
	is first introduced in Ref.~\cite{sugimoto2021test} and
	 shall be studied in this article. The main parameter in the definition is a tunable range $\ell \le L$ of interactions, which allows us to consider how generically the ETH holds in realistic situations.
	\begin{Def}[Hamiltonian] \label{def:model}
		Let $T=T_{L}$ be the (left) translation operator acting on $L$ spins at the vertices of $\mathbb{T}_{L}$.
		We define the ensemble of Hamiltonians with local interactions as
		\begin{equation}
			H_{L}^{(\ell)}
			\coloneqq \sum_{j=0}^{L-1} T_{L}^{-j} \qty( h_{\ell} \otimes I_{L-\ell }  ) T_{L}^{j}\,, \quad
			h_{\ell} \coloneqq \sum_{p_{1}, \ldots ,p_{\ell} = 0}^{3} J_{p_{1},\ldots, p_{\ell}} \sigma_{1}^{ (p_{1}) } \ldots \sigma_{\ell}^{ (p_{\ell}) }
			\label{def_Hamiltonian}
		\end{equation}
		where $\ell \leq L$ is the interaction range, $I_{L-\ell}$ is the identity on the sites $\ell+1,\ldots, L$.
		Here $\sigma^{(p)}_{j}$ is the $p^{\mathrm{th}}$ Pauli matrix $\sigma^{(p)}$ acting on the site $j \in \mathbb{T}_L$, where
		we recall the
		%$\sigma^{(p)}$, $p\in \{0,1,2,3\}$, are the
		standard Pauli matrices,
		\begin{equation} \label{eq:Pauli}
			\sigma^{(0)} = \begin{pmatrix}
				1 & 0 \\ 0 & 1
			\end{pmatrix} \quad
			\sigma^{(1)} = \begin{pmatrix}
				0 & 1 \\ 1 & 0
			\end{pmatrix}
			\quad
			\sigma^{(2)} = \begin{pmatrix}
				0 & -\I \\ \I & 0
			\end{pmatrix} \quad
			\sigma^{(3)} = \begin{pmatrix}
				1 & 0 \\ 0 & -1
			\end{pmatrix} \,.
		\end{equation}

		The $4^\ell$ coefficients $J_{p_{1},\ldots, p_{\ell}}$ are independent, identically distributed real
		Gaussian random variables with zero mean, $\mathbb{E} J_{p_{1},\ldots, p_{\ell}} = 0$, and variance
		$$ v_\ell^2\coloneqq \mathbb{E} \abs*{ J_{p_{1},\ldots, p_{\ell}} }^2 \,.
		$$

	\end{Def}

	The ensemble of Hamiltonians $h_\ell$~\eqref{def_Hamiltonian} contains prototypical spin
	 models such as
	the XYZ model, $h_{\ell} = \sum_{p=1}^{3} J_{p} \sigma_{1}^{(p)} \sigma_{2}^{(p)} $.

	Observe that the Hamiltonian $H_{L}^{(\ell)}$ is a shifted version of the same local
	Hamiltonian $h_\ell$. In particular,
	$H_{L}^{(\ell)}$ is translation invariant by construction,  i.e., $T_{L} H_{L}^{(\ell)} T_{L}^{-1} = H_{L}^{(\ell)}$.
	We impose this structure to study a Hamiltonian with a symmetry.
	In the sequel we shall exploit this feature of $H_L^{(\ell)}$ by switching  from position space to momentum space.

	\begin{lemma}\label{lemma_BlockDiagonalization}
		Let
		\begin{equation} \label{eq:PikTj}
			\Pi_{k} \coloneqq \frac{1}{L} \sum_{j=1}^{L} e^{ 2\pi i \frac{kj}{L} } T_{L}^{-j} \quad \text{for} \quad  k=0,\ldots,L-1
		\end{equation}
		be the projection operator onto the $k$-momentum space, i.e., $T_{L} \Pi_{k} = e^{ 2\pi i \frac{k}{L} } \Pi_{k}$.
		Then $H_{L}^{(\ell)}$ is block-diagonal in the momentum space representation, i.e. in the eigenbasis of $T_L$, since we have
		\begin{equation}
			H_{L}^{(\ell)} = L \sum_{k=0}^{L-1} \Pi_{k} \qty( h_{\ell} \otimes I_{L-\ell}  ) \Pi_{k}\,.
			\label{eq_ExpressionWithProjection}
		\end{equation}
	\end{lemma}
	\begin{proof}
		This follows immediately by substituting the spectral decomposition of $T$ given by $T = \sum_{k=0}^{L-1} e^{ 2\pi i k/L } \Pi_{k}$ into \eqref{def_Hamiltonian}.
	\end{proof}
	As we will show in Lemma \ref{lemma_DimensionOfMomentumSectors}, the dimensions of each of
	the $L$ momentum sectors are almost equal to each other, $\tr_L \Pi_k \approx 2^L/L$.

	In order to present our main result, the ETH in translation-invariant systems~(Theorem~\ref{thm_mainTheoerem}), in a concise form,
	we need to introduce the \textit{microcanonical average}. %(see also Remark \ref{rmk:mcphysics}).
	Below, we denote by $\ket*{ E_{\alpha}^{(k)} }$ the normalised eigenvector of $H_{L}^{(\ell)}$
	in the $k$-momentum sector with eigenvalue $E_{\alpha}$,
	i.e.~$H_L^{(\ell)}\ket*{ E_{\alpha}^{(k)} } = E_\alpha \ket*{ E_{\alpha}^{(k)} }$ and $\Pi_k \ket*{ E_{\alpha}^{(k)} } = \ket*{ E_{\alpha}^{(k)} }$.
	It is easy to see that the spectrum of $H$ in each momentum sector is simple almost surely.

	\begin{Def} [Microcanonical ensemble] \label{def:mcensemble}
		For every energy $E\in \mathbb{R}$ and energy window $\Delta>0$,
		we define the \emph{microcanonical energy shell} $\mathcal{H}_{E,\Delta}$ centered
		at energy $E$ with width $2\Delta$ by
		\begin{equation*}
			\mathcal{H}_{E,\Delta} \coloneqq \bigoplus_{k=0}^{L-1} \mathcal{H}_{E,\Delta}^{(k)}\,, \quad \text{where} \quad   \mathcal{H}_{E,\Delta}^{(k)}
			\coloneqq \mathrm{span}\Bqty{ \ket*{ E_{\alpha}^{(k)} } : \abs{ E_{\alpha}^{(k)} - E } \leq \Delta  }\,.
		\end{equation*}
		We denote the dimension of $\mathcal{H}_{E,\Delta}^{(k)}$ by $d_{E,\Delta}^{(k)}$ and that of  $\mathcal{H}_{E,\Delta}$ by $d_{E,\Delta} = \sum_{k=0}^{L-1}d_{E,\Delta}^{(k)}$.

		Whenever $d_{E, \Delta} > 0$, we define the \emph{microcanonical average}
		of any self-adjoint observable $A \in \mathbb{C}^{N \times N}$  within $\mathcal{H}_{E,\Delta}$ by
		\begin{equation} \label{eq:mcav}
			\expval*{ A }^{(\mathrm{mc})}_{\Delta}(E)
			\coloneqq \frac{1}{ d_{E,\Delta} } \sum_{k=0}^{L-1} \sum_{ \ket*{ E_{\alpha}^{(k)} } \in \mathcal{H}_{E,\Delta}^{(k)}  }
			\expval*{A}{ E_{\alpha}^{(k)} }\,.
		\end{equation}
	\end{Def}
	\begin{remark} \label{rmk:mcphysics}
		The microcanonical average mimics the microcanonical ensemble before taking the thermodynamic limit.
		In order to be physically meaningful, there are two natural requirements on the energy shell
		$\mathcal{H}_{E, \Delta}$:
		\begin{itemize}
			\item[(i)] The density of states is approximately constant in the interval $[E-\Delta, E+\Delta]$.
			\item[(ii)] The microcanonical energy shell contains $\gg1$ states, i.e.~$d_{E, \Delta} \to \infty$ as $L \to \infty$.
		\end{itemize}
		Note that for any fixed energy $E$, (i) corresponds to an upper bound and (ii) corresponds to a lower bound on $\Delta$, both being dependent on $E$.
		We point out that very close to the spectral edges with only a few states, % in particular due to this dependence,
		it is not guaranteed that both requirements can be satisfied simultaneously.

		Indeed, from a physics perspective, viewing $\expval*{ A }^{(\mathrm{mc})}_{\Delta}(E)$ from \eqref{eq:mcav}
		as a finite dimensional approximation of the microcanonical ensemble
		is meaningless whenever (i) and (ii) are not satisfied. However, we will simply view Definition \ref{def:mcensemble} for arbitrary $\Delta$ as an extension of the standard definition of the microcanonical average from the physics literature.
		Our main result, Theorem~\ref{thm_mainTheoerem}, will even hold with the microcanonical average in this extended sense.
	\end{remark}
	We set
	$$N\coloneqq 2^L=\mbox{dim}\mathcal{H}
	$$
	for the total Hilbert space dimension.  Our analytic results below will always be understood in the limit of large system size, i.e. $L\to \infty$, or, equivalently $N\to \infty$. We shall also use the following common notion (see, e.g., \cite{erdos2013local})
		of stochastic domination.
	\begin{Def}    \label{def_StochasticDomination}
		Given two families of non-negative random variables
		\begin{equation*}
			X \coloneqq \qty( X^{(N)}(u) : N \in \mathbb{N},\ u\in U^{(N)} )
			\qq{and}
			Y \coloneqq \qty( Y^{(N)}(u) : N \in \mathbb{N},\ u\in U^{(N)} )
		\end{equation*}
		indexed by $N$, we say that $X$ is stochastically dominated by $Y$, if for all $\xi$, $D > 0$, we have
		\begin{equation*}
			\sup_{u \in U^{(N)}}
			\mathbb{P}\bqty{ X^{(N)}(u) > N^{\xi} Y^{(N)}(u) }
			\leq N^{-D}
		\end{equation*}
		for any sufficiently large $N \ge N_{0}(\xi, D)$ and use the notation $X \prec Y$ or $X =  \mathcal{O}_\prec(Y)$ in that case.
	\end{Def}

%%%%%%%%%%%%%%%%%%%%%%%%%%%%%%%%%%%%%%%%%%%%%%%%%%
%%%%%%%%%%%%%%%%%%%%%%%%%%%%%%%%%%%%%%%%%%%%%%%%%%
\section{Main result in the mean-field case} \label{sec:result}
	Throughout the entire section, we are in the mean-field case $\ell = L$.
	For any $q\le L$ we also introduce the concept of \textit{$q$-local observables} for self-adjoint operators of the form
	$A = A_{q} \otimes I_{L-q}$, i.e.~$A_q$ is self-adjoint and only acts on the first $q$ sites.

	Our main result in this setting is the following theorem.
	\begin{theorem}[ETH in translation-invariant systems]
		\label{thm_mainTheoerem}
		Let $\ell = L$ and consider the Hamiltonian $H_L^{(L)}$ from \eqref{def_Hamiltonian} with eigenvalues
		$E_{\alpha}^{(k)}$ and associated normalised eigenvectors $\ket*{ E_{\alpha}^{(k)} }$.
		Then, for every $\Delta > 0$ and bounded $q$-local observable $A = A_{q} \otimes I_{L-q}$, %i.e.~$A_q$ only acts on the first $q$ sites
		$\norm*{A} \lesssim 1$, it holds that
		\begin{equation} \label{eq_MainResult}
			\max_{\alpha, \beta} \max_{k,k'}\abs{
				\matrixel*{ E_{\alpha}^{(k)} }{ A }{ E_{\beta}^{(k')} }
				- \delta_{\alpha\beta} \delta_{k,k'} \expval*{ A }^{(\mathrm{mc})}_{\Delta}(E_{\alpha}^{(k)})
			}
			\prec \frac{1}{2^{\min\{L/2, L-q\}}}\,,
		\end{equation}
		where the maxima are taken over all indices labeling the eigenvectors of $H_L^{(L)}$.
		In particular, for $q \le L/2$ the ETH holds with optimal speed of convergence of order $1/\sqrt{N}$.
	\end{theorem}
An extension of Theorem \ref{thm_mainTheoerem} to arbitrary dimension $d \ge 2$ is provided in Theorem~\ref{thm_mainTheoerem_Ddimension_HighDimension} in the Appendix.
	\begin{remark} [Typicality of ETH] Theorem~\ref{thm_mainTheoerem} asserts that for any fixed local observable $A$
		the ETH in the form~\eqref{eq_MainResult} holds with a very high probability, i.e. apart from an event
		of probability $N^{-D}= 2^{-LD}$, for any fixed $D$, see the precise Definition~\ref{def_StochasticDomination}.
		This exceptional event may depend on the observable $A$. However, as long as $q$ is $L$-independent (in fact some
		mild logarithmic increase is allowed), it also holds that
		\begin{equation}        \label{eq_MainResult1}
			\max_{\alpha, \beta} \max_{k,k'} \max_{A} \abs{
				\matrixel*{ E_{\alpha}^{(k)} }{ A }{ E_{\beta}^{(k')} }
				- \delta_{\alpha\beta} \delta_{k,k'} \expval*{ A }^{(\mathrm{mc})}_{\Delta}(E_{\alpha}^{(k)})
			}
			\prec \frac{1}{2^{L/2}}\,,
		\end{equation}
		i.e. we may take the supremum over all bounded $q$-local observables $A$ in~\eqref{eq_MainResult}.
		This extension is a simple consequence of choosing a sufficiently fine grid in the unit ball of the $4^q\times 4^q$ dimensional space of $q$-local observables and taking the union bound. The estimate~\eqref{eq_MainResult1} can be viewed as a very strong form of the typicality of ETH within our class of translation invariant mean field operators $H_L^{(L)}$.
		It asserts that apart from an exceptional set of the coupling constants $J_{p_1, \ldots , p_L}$ the Hamiltonian $H_L^{(L)}$ satisfies the ETH with optimal speed of convergence, uniformly in the entire spectrum and tested against all finite range ($q$-local) observables.
		The exceptional set has exponentially small measure of order $2^{-LD}$ for any $D$ if $L$ is sufficiently large.
	\end{remark}

	In Lemma~\ref{lem:mf} we will see that in the mean--field case the Hamiltonian on each momentum sector
	is a GUE matrix, in particular the density of states of $H$ follows Wigner's semicircle law.
	An elementary calculation shows that the radius of this semicircle is given by
	$$ R\coloneqq 2 \cdot 2^L\sqrt{L} v_L(1+ O(2^{-L}))\,.
	$$
	In light of Remark~\ref{rmk:mcphysics} we also mention that $\expval*{ A }^{(\mathrm{mc})}_{\Delta}(E)$
	in~\eqref{eq_MainResult} can be considered
	as an approximation of the expectation of $A$ in the microcanonical ensemble at energy $\abs*{E} \le R$ if
	\begin{equation}\label{D}
		% \frac{2^L\sqrt{L} v}{N^{2/3}}  =
		\frac{R}{N^{2/3}} \ll \Delta \ll R- \abs*{E} \,.
		%2^{2L/3}\sqrt{L} v \ll \Delta \ll 2^L\sqrt{L} v
	\end{equation}
	The upper bound in~\eqref{D} comes from requirement (i) in Remark~\ref{rmk:mcphysics}, while the lower bound in~\eqref{D} stems from (ii) using that the eigenvalue spacing near the spectral edge for Wigner matrices is of order $R/N^{2/3}$.

	\medskip

	For the sequel we introduce the notation
	$$
	\expval*{ A } \coloneqq  \frac{\tr A}{\tr I}
	$$
	for the normalised trace of an operator $A$ on any finite-dimensional Hilbert space, where $I$ is the identity on that space. In particular, if $A = A_{q} \otimes I_{L-q}$ is a $q$-local observable, then $\expval*{ A } = \expval*{ A_{q} }$.

	The proof of Theorem \ref{thm_mainTheoerem} crucially relies
	on the fact that in our mean-field case  $\matrixel*{ E_{\alpha}^{(k)} }{ A }{ E_{\beta}^{(k')} }$ converges
	to $ \delta_{\alpha\beta} \delta_{k,k'} \expval*{ A }$. In other words, the thermodynamics
	of the system is trivial; the thermal value of $A$ is always given by its average trace.  This is formalised in
	the following main proposition:

	\begin{prop} \label{prop:main}
		Under the assumptions of Theorem \ref{thm_mainTheoerem} it holds that
		\begin{equation}        \label{eq:inputProp}
			\max_{\alpha, \beta} \max_{k,k'} \abs{\matrixel*{ E_{\alpha}^{(k)} }{ A }{ E_{\beta}^{(k')} } - \delta_{\alpha\beta} \delta_{k,k'} \expval*{ A }} \prec \frac{1}{2^{\min\{L/2, L-q\}}}\,.
		\end{equation}
	\end{prop}
	Having Proposition \ref{prop:main} at hand, we can readily prove Theorem \ref{thm_mainTheoerem}.
	\begin{proof}[Proof of Theorem \ref{thm_mainTheoerem}]
		Averaging \eqref{eq:inputProp} for $\alpha = \beta$ and $k = k'$ according to the microcanonical average \eqref{eq:mcav}, we find that
		\begin{equation*}
			\max_{\alpha}\max_k \abs{\expval*{ A }^{(\mathrm{mc})}_{\Delta}(E_{\alpha}^{(k)}) -  \expval*{ A }} \prec \frac{1}{2^{\min\{L/2, L-q\}}}\,.
		\end{equation*}
		Combining this with \eqref{eq:inputProp}, the claim immediately follows.
	\end{proof}

	The rest of this section is devoted to the proof of Proposition \ref{prop:main}, which is conducted in four steps.

	\begin{enumerate}
		\item The momentum sectors are all of the same size with very high precision (Lemma~\ref{lemma_DimensionOfMomentumSectors}).
		\item In each momentum sector the  mean-field Hamiltonian $H_L^{(L)}$,
		represented in the eigenbasis of the translation operator $T$, is a GUE matrix (Lemma \ref{lem:mf}).
		\item The ETH holds within each momentum sector separately~(Lemma~\ref{lemma_ETHwithinMomentumSector}).
		\item The averaged trace on each momentum sector and the total averaged trace are close to each other -- at least for local observables~(Lemma~\ref{lemma_DifferenceBetweenSectors}).
	\end{enumerate}

	We shall first formulate all the four lemmas precisely and afterwards conclude the proof of Proposition \ref{prop:main}.
	\begin{lemma}[Step 1: Dimensions of momentum sectors]
		\label{lemma_DimensionOfMomentumSectors}
		The dimension $\tr_{L} \Pi_{k}$ of the $k$-momentum sectors $(k=0,\ldots, L-1)$ is almost equal to each other in the sense that we have
		\begin{equation*}
			\tr_{L} \Pi_{k} = \frac{2^{L}}{L} + \order{ L^{1/2} 2^{L/2} }\,.
		\end{equation*}
	\end{lemma}
	The proof is given in Section \ref{subsec:dimmomsec}

	\begin{lemma}[Step 2: GUE in momentum blocks]\label{lem:mf}
		Each momentum-block of the mean-field Hamiltonian $H_{L}^{(L)}$, represented in an eigenbasis of $T$, is an i.i.d.~complex Gaussian Wigner matrix (GUE), whose entries have mean zero and variance
		$2^L \, L^2 \, v_L^2$.    Recall that $v_L^2 = \mathbb{E} \abs*{ J_{p_{1},\ldots, p_{L}} }^2$ from Definition \ref{def:model}.
	\end{lemma}
	\begin{proof}
		In the mean-field case $\ell=L$,  a simple direct calculation of all  first and second moments
		of the matrix elements
		shows that
		the interaction matrix $h_{\ell}$ is a complex Gaussian Wigner matrix whose entries have variance $2^L v_L^2$.
		Since the transformation from the standard basis to an eigenbasis of $T$ is unitary, and the Gaussian distribution is invariant under unitary transformation, $h_{\ell}$ represented in an eigenbasis of $T$ is again a Gaussian Wigner matrix.
		Finally, the projection operators $\Pi_{k}$ in \eqref{eq_ExpressionWithProjection} set the off-diagonal blocks to zero. Incorporating the additional factor $L$ in \eqref{eq_ExpressionWithProjection} into the variance proves Lemma~\ref{lem:mf}.
	\end{proof}

	As the next step, we show that the ETH holds within each momentum sector.

	\begin{lemma}[Step 3: ETH within each momentum sector]
		\label{lemma_ETHwithinMomentumSector}
		For an  arbitrary deterministic observable $A$ with $\norm*{A} \lesssim 1$ it holds that
		\begin{equation}\label{ethsector}
			\max_{\alpha, \beta}\max_{k}    \abs{ \matrixel*{ E_{\alpha}^{(k)} }{ A }{ E_{\beta}^{(k)} } - \delta_{\alpha\beta} \frac{ \tr_{L}\qty( \Pi_{k} A \Pi_{k} ) }{ \tr_{L} \Pi_{k} } }
			\prec
			\frac{1}{2^{L/2} }\,.
		\end{equation}
	\end{lemma}
	\begin{proof}
		For any fixed $k$, Lemma~\ref{lem:mf} asserts that
		$\Pi_k H_L^{(L)}\Pi_k $ is a standard GUE matrix (up to normalisation by $v_L$).
		Using~\cite[Theorem 2.2]{cipolloni2021eigenstate},
		therefore its eigenvectors  $\ket*{E_\alpha} = \ket*{E_\alpha^{(k)}}$
		satisfy ETH in the form that $\matrixel*{ E_{\alpha} }{ A }{ E_{\beta} }$
		is approximately given by the normalised trace of $A$ in the $k$-momentum sector
		$$
		\expval*{A}_k \coloneqq \frac{ \tr_{L}\qty( \Pi_{k} A \Pi_{k} ) }{ \tr_{L} \Pi_{k} }
		$$
		with very high probability and with
		an error given by the square root of the inverse of the dimension of the
		$k$-momentum sector,  $1/\sqrt{ \tr_L\Pi_k}$. This holds in the sense
		of stochastic domination given in Definition~\ref{def_StochasticDomination}.
		Using that  $\tr_L\Pi_k\approx 2^L/L$ from Lemma~\ref{lemma_DimensionOfMomentumSectors},
		we obtain that~\eqref{ethsector} holds for each fixed $k$, uniformly in all eigenvectors.
		Finally, the very high probability control in the stochastic domination allows us to take the
		maximum over $k=1,2,\ldots, L$ by a simple union bound.  This completes the proof
		of~\eqref{ethsector}.
	\end{proof}

	We remark that the essential ingredient of this proof, the Theorem  2.2 from \cite{cipolloni2021eigenstate},
	applies not only for the Gaussian ensemble but for arbitrary Wigner matrices with i.i.d. entries
	(with some moment condition on their entry distribution)
	and its proof is quite involved.
	However, ETH for GUE, as needed in Lemma~\ref{lemma_ETHwithinMomentumSector}, can also be proven with much more elementary methods using that the eigenvectors are columns of a Haar unitary matrix.
	Namely, moments of
	$\matrixel*{ E_{\alpha} }{ A }{ E_{\beta} } $ can be directly
	computed using Weingarten calculus~\cite{collins2022weingarten}.
	Since in~\eqref{ethsector} we aim
	at a control with very high probability, this would require to compute arbitrary high moments
	of $\matrixel*{ E_{\alpha} }{ A }{ E_{\beta} } -\delta_{\alpha, \beta} \expval*{A}_k$.
	The Weingarten formalism gives the exact answer but it is somewhat complicated
	for high moments, so identifying their leading order (given by the ``ladder'' diagrams)
	requires some elementary efforts. For brevity, we therefore relied on~\cite[Theorem 2.2]{cipolloni2021eigenstate} in the
	proof of Lemma~\ref{lemma_ETHwithinMomentumSector} above.

	\medskip

	Finally, we formulate the fourth and last
	step of the proof of Proposition \ref{prop:main} in the following lemma, the
	proof of which is given in  Section~\ref{sect_DifferenceBetweenSectors}.
	\begin{lemma}[Step 4: Traces within momentum sectors]
		\label{lemma_DifferenceBetweenSectors}
		Let $A = A_{q} \otimes I_{L-q}$ be an arbitrary $q$-local observable with $\norm*{ A } \lesssim 1$.
		Then it holds that
		\begin{equation}
			\max_k \abs{ \frac{ \tr_{L}\qty( \Pi_{k} A \Pi_{k} ) }{ \tr_{L} \Pi_{k} } - \expval*{ A } }
			\leq \order{ \frac{ L }{ 2^{\min\Bqty{L-q, L/2}} } }\,.
			\label{eq_DifferenceBetweenSectors}
		\end{equation}
		Moreover, for $q >L/2+1$ this bound is optimal (up to the factor $L$).
	\end{lemma}
	Armed with these four lemmas, we can now turn to the proof of Proposition \ref{prop:main}.

	\begin{proof}[Proof of Proposition~\ref{prop:main}]

		First, for any $q$-local observable $A = A_{q} \otimes I_{L-q}$, we conclude from Lemma \ref{lemma_ETHwithinMomentumSector} and Lemma~\ref{lemma_DifferenceBetweenSectors} that
		\begin{equation} \label{eq:blockETH}
			\max_{\alpha, \beta}\max_{k}    \abs{ \matrixel*{ E_{\alpha}^{(k)} }{ A }{ E_{\beta}^{(k)} } - \delta_{\alpha\beta} \frac{ \tr_{L}\qty( \Pi_{k} A \Pi_{k} ) }{ \tr_{L} \Pi_{k} } }\prec \frac{1}{ 2^{\min\Bqty{L-q, L/2}} }\,.
		\end{equation}

		For the element $\matrixel*{ E_{\alpha}^{(k)} }{ A }{ E_{\beta}^{(k')} }$ with $k\neq k'$, i.e.~in off-diagonal blocks, $\ket*{ E_{\alpha}^{(k)} }$ and $\ket*{ E_{\beta}^{(k')} }$ are normalised Gaussian vectors independent of each other.
		Therefore standard concentration estimate shows that
		\begin{align}         \label{eq:betweenblockETH}
			\max_{k\neq k'}\abs{
				\matrixel*{ E_{\alpha}^{(k)} }{ A }{ E_{\beta}^{(k')} }
			}
			\prec \frac{1}{ 2^{L/2} }\,.
		\end{align}
		Combining \eqref{eq:blockETH} with \eqref{eq:betweenblockETH}, we have proven Proposition \ref{prop:main}.
	\end{proof}

	\subsection{Dimensions of momentum sectors: Proof of Lemma \ref{lemma_DimensionOfMomentumSectors}} \label{subsec:dimmomsec}
		In this section we prove Lemma \ref{lemma_DimensionOfMomentumSectors}, and establish that the sizes of the momentum sectors are almost equal.
		To this end, we show that the leading term in the size of each of the momentum blocks is given by the number of aperiodic elements in the product basis of $\mathcal{H}$.

		We present the proof using group theory notation, which is not strictly necessary for the one-dimensional case under consideration since the translation group of the torus $\mathbb{T}_{L}$ is cyclic. Nevertheless, we do it to allow for a more straightforward generalisation to the $d$-dimensional case (cf.~Lemma \ref{lemma_DimensionOfMomentumSectors_HighDimension}).
		\begin{proof}
			We introduce the following objects.
			Let $\s$ denote the canonical product basis of $\mathcal{H}$,
			\begin{equation} \label{eq:purestates}
				\s(L) \coloneqq \{ \sigma:\mathbb{T}_{L} \to \{\ket{\uparrow},\ket{\downarrow}\}\}\,,
			\end{equation}
			and let $\G$ be the group of translations of $\mathbb{T}_{L}$ generated by $T=T_L$. Note that $\G$ is a finite cyclic group of size $\abs*{\G} = L$.
			The action of $\G$ on $\s(L)$ is defined by
			\begin{equation} \label{eq:gactsons}
				(g\sigma)(x) \coloneqq \sigma(g^{-1}(x))\,,\quad x\in\mathbb{T}_{L}\,, \quad \sigma \in \s(L)\,, \quad g\in\G\,.
			\end{equation}
			In particular, the set $\s(L)$ is a disjoint union of sets $\s_b(L)$ defined by
			\begin{equation*} \label{eq:s_b}
				\s_b(L) \coloneqq \{\sigma \in \s(L) \,:\, \abs*{\G_\sigma}  = b \}\,, \quad b = 1,2,\ldots, L\,,
			\end{equation*}
			where $\G_\sigma \subset \s(L)$ is the stabilizer of $\sigma$ under the action \eqref{eq:gactsons}.
			By the orbit-stabilizer theorem, $\s_b(L) = \emptyset$ for all $b$ that do not divide $L$.
			Since the group $\G$ is cyclic, it has a unique subgroup of size $b$ for all $b\vert L$, given explicitly by
			\begin{equation*} \label{eq:G_b}
				\G^{(b)} \coloneqq \{T^{L/b},T^{2L/b},\ldots, T^{L}\}\,.
			\end{equation*}

			Observe that each $\sigma \in \s_b(L)$ corresponds to a unique
			map $\widetilde{\sigma}$ on a reduced torus
			$\s(L/b) \coloneqq \faktor{\mathbb{T}_L}{\G^{(b)}}$, which is defined by
			\begin{equation} \label{eq:reduced_sigma}
				\widetilde{\sigma}([x]) \coloneqq \sigma(x)\,, \quad [x] \in \s(L/b)\,.
			\end{equation}
			Since $\sigma$ is stabilised by $\G^{(b)}$,
			the map $\sigma \mapsto \widetilde{\sigma}$ in \eqref{eq:reduced_sigma}
			is well-defined and injective.
			In particular, $\abs*{ \s_b(L) }  \le 2^{L/b}$, and hence
			\begin{equation} \label{eq:Midentity}
				2^L = \sum\limits_{b\vert L} \abs*{ \s_b(L) }  =
				M(L) + \sum\limits_{b\vert L, b\ge 2} \abs*{ \s_b(L) }
				\le M(L) + \order{L^{1/2}2^{L/2}}\,,
			\end{equation}
			where $M(L) \coloneqq \abs*{ \s_1(L) } $ denotes the number of elements in $\s(L)$ with a trivial stabilizer. The last inequality follows from the fact that $L$ has at most $\mathcal{O}(L^{1/2})$ divisors.

			Since $M(L) \le 2^L$, we conclude from \eqref{eq:Midentity} that
			\begin{equation} \label{eq:Mbound}
				M(L) = 2^L + \order{L^{1/2}2^{L/2}}\,.
			\end{equation}

			For any $k \in \{0,\ldots,L-1\}$, we can construct an eigenvector of $T$ corresponding to the eigenvalue $e^{ 2\pi i k/L }$ by defining
			\begin{equation} \label{eq:vsigmadef}
				\vb{v}(\sigma,k) \coloneqq \Pi_{k}\sigma = \frac{1}{L}\sum\limits_{j=0}^{L-1} e^{ 2\pi i \frac{kj}{L} } T^{-j}\sigma\,,  \quad \sigma \in \s_1(L)\,.
			\end{equation}
			Since the orbit of $\sigma$ under $T$ consists of $L$ distinct basis elements, the vector $\vb{v}(\sigma,k)$ is non-zero.
			Furthermore, the vectors $\vb{v}(\sigma,k)$ and $\vb{v}(\sigma',k)$ corresponding to $\sigma$ and $\sigma'$ in disjoint orbits are linearly independent because they share no basis element.
			Therefore, the dimension of the $k$-th momentum space is bounded from below by the number of disjoint orbits in $\s_1(L)$, that is
			\begin{equation} \label{eq:tr_lower_bound}
				\tr_{L} \Pi_{k} \ge \frac{2^{L}}{L} + \order{L^{-1/2} 2^{L/2} }\,,
			\end{equation}
			where we used inequality \eqref{eq:Mbound} and the fact that all orbits in $\s_1(L)$ have size $L$. By means of \eqref{eq:tr_lower_bound}, we obtain the following chain of inequalities
			\begin{equation} \label{eq:tr_upper_bound}
				\tr_{L}\Pi_{k} = 2^L - \sum\limits_{j\neq k}\tr_{L}\Pi_{j} \le \frac{2^{L}}{L} + \order{L^{1/2} 2^{L/2} }\,,
			\end{equation}
			which, together with \eqref{eq:tr_lower_bound} concludes the proof of Lemma~\ref{lemma_DimensionOfMomentumSectors}.
		\end{proof}

	\subsection{Traces within momentum sectors: Proof of Lemma \ref{lemma_DifferenceBetweenSectors}}
		\label{sect_DifferenceBetweenSectors}
		In this section, we give a proof of Lemma~\ref{lemma_DifferenceBetweenSectors}, which evaluates the difference of the noramalised trace $\tr_{L}( \Pi_{k} A \Pi_{k})/\tr_L \Pi_k$ on a momentum sector and the full normalised trace $\expval*{A}$ for a $q$-local observable $A = A_{q} \otimes I_{L-q}$.
		We separate $A$ into the tracial part $\expval*{ A } I$ and the traceless part $\mathring{A} \coloneqq A - \expval*{ A } I$.

		\begin{proof}[Proof of Lemma~\ref{lemma_DifferenceBetweenSectors}]
			Substituting $\Pi_{k} \coloneqq \frac{1}{L} \sum_{j=1}^{L} e^{ 2\pi i \frac{kj}{L} } T_{L}^{-j}$, we obtain
			\begin{align}
				\tr_{L}( \Pi_{k} A \Pi_{k})
				&= \expval*{ A } \tr_{L}\Pi_{k} + \tr_{L}( \Pi_{k} \mathring{A} ) \nonumber \\
				&= \expval*{ A } \tr_{L}\Pi_{k} + \frac{1}{L} \sum_{j=1}^{L-1} e^{ 2\pi i \frac{kj}{L} } \tr_{L}( T_{L}^{-j} \mathring{A} )\,.
				\label{TraceWithinMomentumSector}
			\end{align}
			Then, the task is to evaluate the size of the quantity $\tr_{L}( T_{L}^{-j} \mathring{A} )$.
			\begin{lemma}
				\label{lemma_BoundForTraceWithTranslation}
				Let $A \coloneqq A_{q} \otimes I_{L-q}$ be a $q$-local observable with $\norm*{ A } \lesssim 1$.
				Then, for any $j=1,\ldots, L-1$, we have
				\begin{equation}
					\abs{ \tr_{L}( T_{L}^{-j} A ) }
					\lesssim 2^{ \max\Bqty{q, \gcd(j, L)} }\,,
					\label{eq_BoundForTraceWithTranslation}
				\end{equation}
				where $\gcd$ stands for the greatest common divisor.
			\end{lemma}
			Combining \eqref{eq_BoundForTraceWithTranslation} with $\gcd(j, L) \leq L/2$ for $j=1,\ldots,L-1$ and Lemma~\ref{lemma_DimensionOfMomentumSectors} gives the bound~\eqref{eq_DifferenceBetweenSectors}. The optimality of \eqref{eq_DifferenceBetweenSectors} for $q >L/2+1$ is proven in Lemma \ref{lemma_PartialTraceForBl} below.
		\end{proof}

		It remains to give the proof of Lemma \ref{lemma_BoundForTraceWithTranslation}.

		\begin{proof}[Proof of Lemma \ref{lemma_BoundForTraceWithTranslation}]
			We choose a product basis $\Bqty*{ \ket*{ s_{1}\ldots s_{L} } \mid s_{j} \in \Bqty*{ \up,\dn } }$ to calculate the trace.
			Then, we obtain
			\begin{align}
				\abs{ \tr_{L}( T_{L}^{-j} A ) }
				&= \abs{ \sum_{ s_{1} \ldots s_{L} } \matrixel*{ s_{1+j} \ldots s_{q+j} }{ A_{q} }{ s_{1} \ldots s_{q} } \prod_{m=q+1}^{L} \delta_{ s_{m} s_{m+j} } }
				\nonumber \\
				&\lesssim \sum_{ s_{1} \ldots s_{L} } \prod_{m=q+1}^{L} \delta_{ s_{m} s_{m+j} }\,.
				\label{j_term}
			\end{align}

			Because of the product $\prod_{m=q+1}^{L} \delta_{ s_{m} s_{m+j} }$ of Kronecker deltas, not all of the summation variables $s_{1},\ldots, s_{L}$ are independent.

			To count the number of independent summations in the right-hand side of \eqref{j_term} and obtain an upper bound for $\tr_{L}( T_{L}^{-j} A )$ with $j=1,\ldots, L-1$, we count the number of independent deltas in the product
			\begin{equation} \label{eq:Gdef}
				\mathcal{G}_{q,j}^{(L)} \coloneqq \prod_{m=q+1}^{L} \delta_{ s_{m} s_{m+j} }\,.
			\end{equation}
			Here, not all of the delta functions in $\mathcal{G}_{q,j}^{(L)}$ are independent in the sense that we may express $\mathcal{G}_{q,j}^{(L)}$ with a fewer number of deltas.
			For example, we have $\mathcal{G}_{1,2}^{(4)} = \delta_{s_{2} s_{4}} \delta_{s_{3} s_{1}} \delta_{s_{4} s_{2}} = \delta_{s_{3} s_{1}} \delta_{s_{4} s_{2}}$.

			To obtain an expression of $\mathcal{G}_{q,j}^{(L)}$ with the minimal number of deltas,
			we graphically represent the product $\prod_{m=q+1}^{L} \delta_{ s_{m} s_{m+j} }$ by arranging the sites
			on a circle and representing the $\delta_{s_{m} s_{m+j}}$'s with a line connecting
			the site $m$ and $m+j$~(Figure~\ref{fig_GraphicalRepresentationOfProductOfDelta}).
			A minimal representation of $\mathcal{G}_{q,j}^{(L)}$ is obtained by removing exactly one
			delta for every occurrence of a loop in the graph of $\prod_{m=q+1}^{L} \delta_{ s_{m} s_{m+j} }$.

			The graph of $\prod_{m=q+1}^{L} \delta_{ s_{m} s_{m+j} }$ can be obtained in two steps: First, in step (i),
			drawing the graph of $\prod_{m=1}^{L} \delta_{ s_{m} s_{m+j} }$
			and second, in step (ii), removing the lines corresponding to the delta functions
			$\delta_{s_{m} s_{m+j}} \ (m=1,\ldots,q)$, which are depicted with red dashed lines in Figure~\ref{fig_GraphicalRepresentationOfProductOfDelta}.

			In the first step~(i), there are exactly $\gcd(j, L)$ loops each starting from the sites $1,\ldots, \gcd(j, L)$.
			If $q>\gcd(j, L)$, there is no loop remaining after the second step~(ii).
			Thus, we obtain a minimal representation of $\mathcal{G}_{q,j}^{(L)}$ as $\mathcal{G}_{q,j}^{(L)} = \prod_{m = q + 1}^{L} \delta_{ s_{m} s_{m+j} }$.
			If $q\leq\gcd(j, L)$, the loops starting from the sites $q+1,\ldots,\gcd(j, L)$ remain after the second step~(ii), for each of which we remove one delta to obtain a minimal representation of $\mathcal{G}_{q,j}^{(L)}$ as $\mathcal{G}_{q,j}^{(L)} = \prod_{m = \gcd(j, L) + 1}^{L} \delta_{ s_{m} s_{m+j} }$.

			In summary, we obtain a minimal representation of $\mathcal{G}_{q,j}^{(L)}$ as
			\begin{equation}
				\mathcal{G}_{q,j}^{(L)}
				= \prod_{m = \max\Bqty{q, \gcd(j, L)} + 1}^{L} \delta_{ s_{m} s_{m+j} }\,.
				\label{eq_MininalRepresentation}
			\end{equation}
			By substituting \eqref{eq_MininalRepresentation} into \eqref{j_term} we obtain
			\begin{align*}
				\abs{ \tr_{L}( T_{L}^{-j} A ) }
				&\lesssim \sum_{ s_{1} \ldots s_{L} } \prod_{m = \max\Bqty{q, \gcd(j, L)} + 1}^{L} \delta_{ s_{m} s_{m+j} }
				= 2^{ \max\Bqty{q, \gcd(j, L)} }\,. \qedhere
			\end{align*}
		\end{proof}

		\begin{figure}[tbh]
			\centering
			\includegraphics[width=\linewidth]{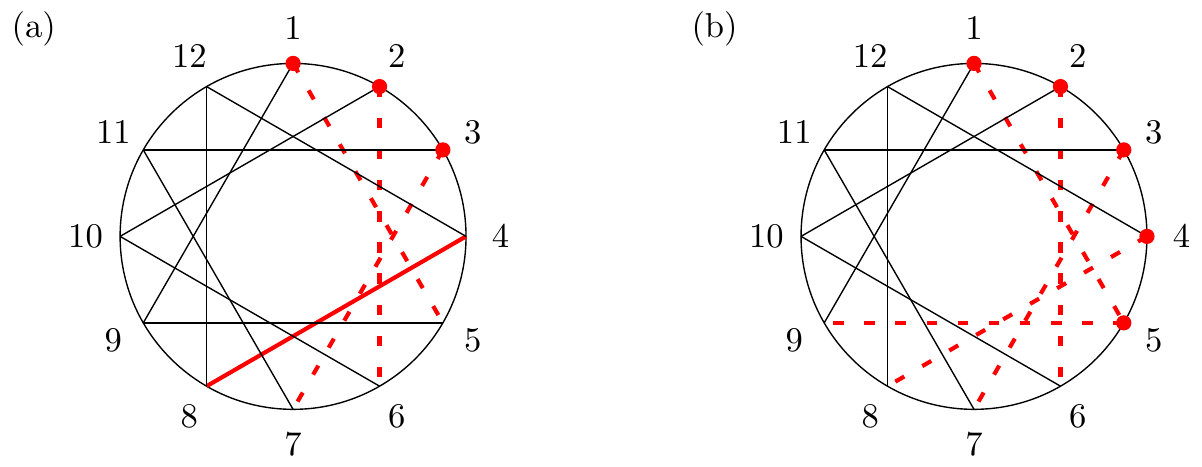}
			\caption{
				Graphical representation of the product $\prod_{m=q+1}^{L} \delta_{ s_{m+j} s_{m} }$ for (a) $L=12$, $q=3$, $j=4$ and (b) $L=12$, $q=5$, $j=4$.
				For the first case~(a) where $q < \gcd(j, L)$, there is a loop $4$-$8$-$12$-$4$ remaining after the step~(ii), which contains exactly one redundant delta function $\delta_{s_{4} s_{8}}$ depicted with a solid red line.
				In general, exactly one redundant delta function appears for every occurrence of a loop in the graph of $\prod_{m=q+1}^{L} \delta_{ s_{m+j} s_{m} }$.
			}
			\label{fig_GraphicalRepresentationOfProductOfDelta}
		\end{figure}

		%%%%%%%%%%%%%%%%%%%%%%%%%%%%%%%%%%%%%%%%%%%%%%%%%%
		%%%%%%%%%%%%%%%%%%%%%%%%%%%%%%%%%%%%%%%%%%%%%%%%%%
		%\subsubsection{Optimality of the bound given in Lemma~\ref{lemma_DifferenceBetweenSectors}}
		Finally, we prove the optimality of~\eqref{eq_DifferenceBetweenSectors} in the regime $q > L/2+1$.
		\begin{lemma}
			\label{lemma_PartialTraceForBl}
			Let $B_{q} \coloneqq T_{q} + T_{q}^{-1} - 2^{2-q} I_{q}$, where $T_q$ is the (left) translation operator acting only on the first $q$ spins arranged on the torus $\mathbb{T}_q$. Observe that $B_q$ is Hermitian and traceless.
			Then, for $q > L/2 + 1$, the normalised trace of $B \coloneqq B_{q} \otimes I_{L-q}$ within the $k$-momentum sector is given by
			\begin{equation}
				\frac{ \tr_{L}( \Pi_{k} B \Pi_{k}) }{ \tr_{L} \Pi_{k} }
				= \frac{2}{ 2^{L-q} } \cos\qty( \frac{2\pi k}{L} ) + \order{ \frac{L}{ 2^{L/2} } }\,.
				\label{ed_PartialTraceForBl}
			\end{equation}
		\end{lemma}
		This shows that the $q$-local observable $B_{q} \coloneqq T_{q} + T_{q}^{-1} - 2^{2-q} I_{q}$ saturates the bound~\eqref{eq_DifferenceBetweenSectors} when $q> L/2 + 1$.
		It also shows that the deviation of the normalised trace within a momentum sector, $\tr_{L}( \Pi_{k} B \Pi_{k}) / \tr_{L} \Pi_{k}$, from $\langle B \rangle = 0$, which is of order $2^{-(L-q)}$, becomes the dominant source of error in the ETH whenever $q> L/2 + 1$.

		\begin{proof}[Proof of Lemma~\ref{lemma_PartialTraceForBl}]
			We first reduce the range of the summation over $j$ in the generally valid expression \eqref{TraceWithinMomentumSector} applied to $B_q$.
			To do so, we introduce the parity operator $P_{L}$ defined by $P_{L} \ket*{s_{1} s_{2} \ldots s_{L}} \coloneqq \ket*{s_{L} \ldots s_{2} s_{1}}$.
			It satisfies $P_{L} T_{L} P_{L} = T_{L}^{-1}$ and $P_{L} A P_{L} = I_{L-q} \otimes (P_{q} A_{q} P_{q})$ for any $A = A_q \otimes I_{L-q}$.
			Since $B_{q}$ is invariant under the parity transformation, we have
			\begin{equation*}
				\tr_{L}( T_{L}^{-j} B )
				= \tr_{L}[ T_{L}^{+j} (I_{L-q} \otimes B_{q}) ]
				= \tr_{L}( T_{L}^{-(L-j)} B )\,,
			\end{equation*}
			Therefore, we can rewrite \eqref{TraceWithinMomentumSector} with the aid of \eqref{eq:PikTj} as
			\begin{align}
				\tr_{L}( \Pi_{k} B \Pi_{k})
				&= \frac{2}{L} \sum_{j=1}^{ \lfloor \frac{L}{2} \rfloor}
				\tr_{L}( T_{L}^{-j} \mathring{B} ) \cos\qty( \frac{2\pi k j}{L} )
				+
				\begin{cases}
					\frac{ (-1)^{k} }{L} \tr_{L}( T_{L}^{-\frac{L}{2}} \mathring{B} ) & L \ \text{even} \\
					0 & L \ \text{odd}\,.
				\end{cases}
				\label{eq_PartialTraceOfB}
			\end{align}

			When $q > L/2 + 1$, we have $j < q$ and cannot skip over the region $1,\ldots,j$ when going along the lines in the graph of $\mathcal{G}_{q,j}^{(L)}$ (recall \eqref{eq:Gdef}).
			Therefore, each line starting at one of the sites $p \in \{q+1,\ldots, q+j\}$ passes through a point in $\{1,\ldots,j\}$.
            Moreover, the correspondence between $p$ and the first intersection of the line starting at $p$ with $\{1,\ldots,j\}$ is one-to-one.
			Therefore, there exists a permutation $\tau_{j}$ on $1,\ldots,j$ such that $s_{ q + i } = s_{ \tau(i) }$ for $i=1,\ldots,j$ due to $\mathcal{G}_{q,j}^{(L)}$.
			With this permutation $\tau$, we obtain
			\begin{align}
				\tr_{L}( T_{L}^{-j} B )
				&= \sum_{s_{1}\ldots s_{L}} \matrixel*{ s_{q+1} \ldots s_{q} s_{q+1} s_{q+j} }{ B_{q} }{ s_{1} \ldots s_{q} } \mathcal{G}_{q,j}^{(L)}
				\nonumber \\
				&= \sum_{s_{1}\ldots s_{q}} \matrixel*{ s_{q+1} \ldots s_{q} s_{\tau(1)} s_{\tau(j)} }{ B_{q} }{ s_{1} \ldots s_{q} }
				\nonumber \\
				&= \tr_{q}( \tau_{j}^{\dagger} T_{q}^{-j} B_{q} )
				\nonumber \\
				&= \tr_{q}( \tau_{j}^{\dagger} T_{q}^{-(j-1)} ) + \tr_{q}( \tau_{j}^{\dagger} T_{q}^{-(j+1)} ) -2^{2-q} \tr_{q}( \tau_{j}^{\dagger} T_{q}^{-j} )\,.
				\label{eq_TraceOfBwithTranslation}
			\end{align}
			Because $\tau_{j}$ is a $j$-local operator (not necessarily self-adjoint) on the $q$-site chain, we can apply Lemma~\ref{lemma_BoundForTraceWithTranslation} to each term in \eqref{eq_TraceOfBwithTranslation}.
			Combined with $j < q-1$ and $\gcd(j,q) \leq j$, we obtain
			\begin{equation*}
				\tr_{L}( T_{L}^{-j} B )
				= \delta_{j1} 2^{q} + \order{ 2^{j} }
				= \delta_{j1} 2^{q} + \order{ 2^{L/2} }\,.
			\end{equation*}
			Substituting this result into \eqref{eq_PartialTraceOfB} and employing $\tr_{L} \Pi_{k} = \frac{ 2^{L} }{L} + \order{ L^{1/2}2^{L/2}}$ from Lemma~\ref{lemma_DimensionOfMomentumSectors}, we obtain the result~\eqref{ed_PartialTraceForBl}.
		\end{proof}

%%%%%%%%%%%%%%%%%%%%%%%%%%%%%%%%%%%%%%%%%%%%%%%%%%
%%%%%%%%%%%%%%%%%%%%%%%%%%%%%%%%%%%%%%%%%%%%%%%%%%
\section{Numerical verification of Theorem~\ref{thm_mainTheoerem} for $\ell = \mathcal{O}(1)$} \label{sec:numerics}
	In this section, we numerically demonstrate that Theorem~\ref{thm_mainTheoerem} also holds for the non-mean-field case of $\ell=2$.
	For that purpose,  we adopt the following measure of the ETH used in Refs~\cite{sugimoto2021test, sugimoto2022eigenstate}.  For any self-adjoint operator $A$ we define
	\begin{equation}
		\Lambda =  \Lambda(A)  \coloneqq \mean
		\max_{k} {\max_{\alpha}}^{\prime} \frac{ \abs{ \matrixel*{ E_{\alpha}^{(k)} }{ A }{ E_{\alpha}^{(k)} } - \expval*{ A }^{(\mathrm{mc})}_{\Delta}(E_{\alpha}^{(k)}) } }{ a_{\max} - a_{\min} }\, ,
        \label{eq_NumericsMCShell}
	\end{equation}
        where $a_{\max(\min)}$ is the maximum (minimum) eigenvalue of $A$.
	Here, $\mean$ denotes the average over the realisations of the Hamiltonian~\eqref{def_Hamiltonian}, and ${\max_{\alpha}}^{\prime}$ denotes the maximum over the eigenstates $\ket*{ E_{\alpha}^{(k)} }$ in the energy shell at the center of the spectrum, i.e.~those $\alpha$ for which
	\begin{equation*}
		\abs{ E_{\alpha}^{(k)} - \expval*{H} } \leq \Delta\,.
	\end{equation*}
	The width $\Delta$ of the energy interval is set to be $\Delta = 0.4/L$ such that it satisfies the two physical requirements mentioned in Remark~\ref{rmk:mcphysics} for $L\geq 6$.
	With this choice of $\Delta$, the microcanonical energy shell $\mathcal{H}_{\expval*{H}, \Delta}$ defined by \eqref{eq_NumericsMCShell} typically contains more than 10 states, while the density of states does not change too much within $\mathcal{H}_{\expval*{H}, \Delta}$.

    As the observable, we choose $A = B_{q} \otimes I_{L-q}$ with $B_{q} \coloneqq T_{q} + T_{q}^{-1} - 2^{2-q} I_{q} $ for $q=2,\ldots, L$, which saturates the upper bound in \eqref{eq_DifferenceBetweenSectors} and thus also saturates that of \eqref{eq_MainResult1}.
    With this choice we have $a_{\max} - a_{\min} \simeq 4$ for any $L$ and $q$.
    Therefore,    the ETH measure $\Lambda$ is essentially the same as the diagonal part of the left-hand side of \eqref{eq_MainResult1} in Theorem~\ref{thm_mainTheoerem}  --
    except that the maximum over $\alpha$ is now taken only at the center of the spectrum
  (and we do not take maximum over all $A$).
	This is because the eigenstate expectation value $\expval*{ A }{ E_{\alpha}^{(k)} }$ of a local observable $A = A_{q}\otimes I_{L-q} $ with $q \ll L$ typically acquires an energy dependence when $\ell \ll L$~\cite{hamazaki2018atypicality}, and the number of states becomes not enough to calculate the microcanonical average near the edges for the computationally accessible system size.
    The ETH measure $\Lambda$ satisfies reasonable thermodynamical properties. It is  (i) invariant under the linear transformation $A\mapsto aA +b$, (ii) dimensionless, and (iii) thermodynamically intensive for additive observables $A$~\cite{sugimoto2021test}.

	Figures \ref{fig_ETHMeasure}(a)-(c) depict the $L$-dependence of the ETH measure $\Lambda$ for different values of the parameter $q$.
	In particular, Figure \ref{fig_ETHMeasure}(b) illustrates that, whenever $L$ is approximately equal to $q$ so that $L-q < L/2$, the ETH measure $\Lambda$ decays as $\propto 2^{-L}$.
	The rate of this decay is slower for smaller values of $q$, but approaches $2^{-L}$ as $q$ becomes larger.
	In Figure \ref{fig_ETHMeasure}(c), we take a closer look at the $L$-dependence of $\Lambda$ for $q=6$.
	The data indicates that for $L - q \ll L/2$, $\Lambda$ decays as $\propto 1.8^{-L}$, whereas for $L\gtrsim 2q$, $\Lambda$ decays as $\propto 1.8^{-L/2}$.
	These numerical observations are in agreement with our analytical results for the mean-field case in Theorem~\ref{thm_mainTheoerem}, which predicts that the exponent of the exponential decrease of $\Lambda$ in $L$ should be twice as large in the region $L-q \ll L/2$ compared to the region $L/2 \gtrsim L-q$.
	This fact suggests that the theorem remains qualitatively valid for $\ell=\order{1}$ in the bulk of the spectrum as long as the energy shell width is appropriately chosen.

	\begin{figure}[htb]
		\centering
		\includegraphics[width=0.9\linewidth]{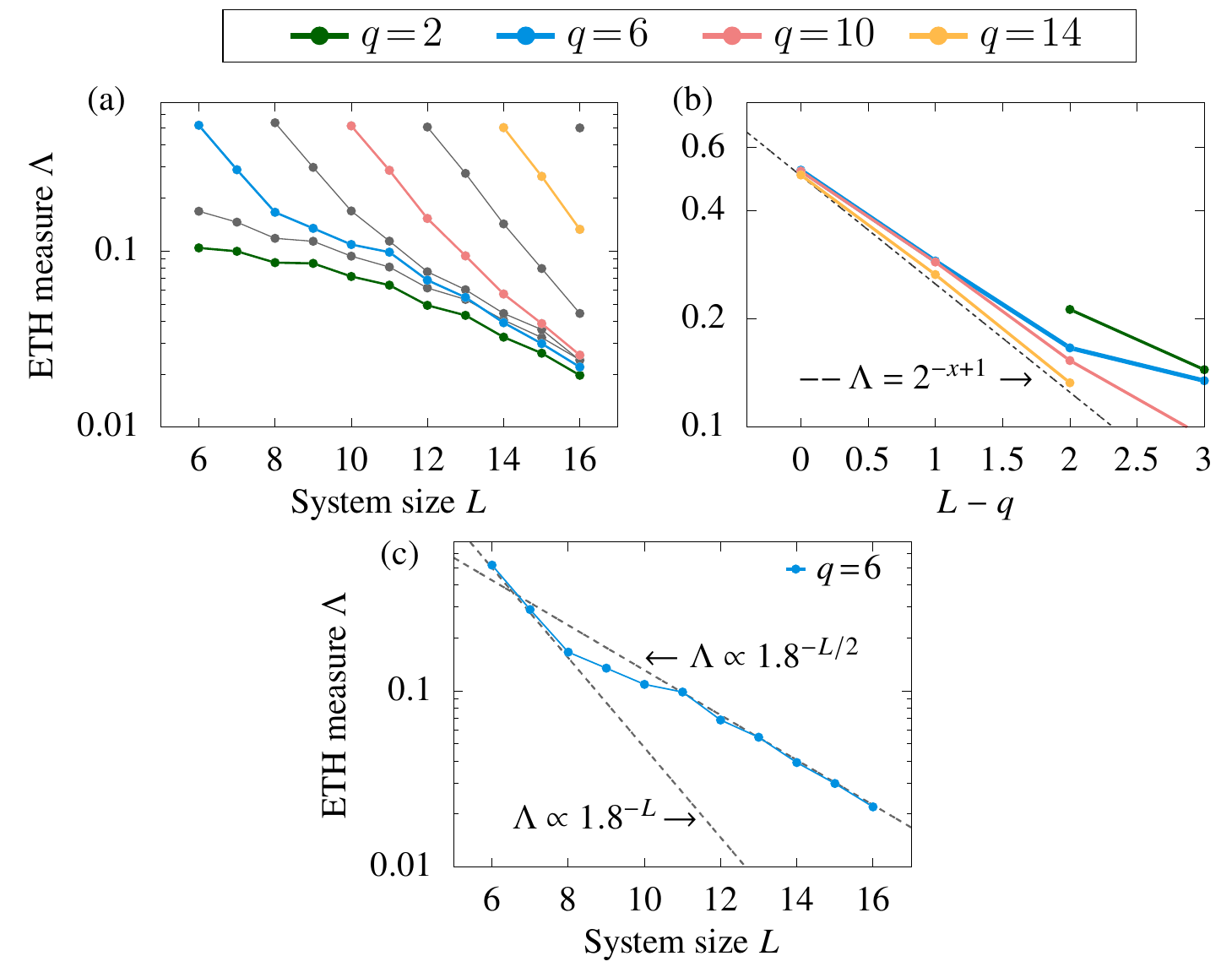}
		\caption{
			\label{fig_ETHMeasure}
			(a) System-size dependence of the ETH measure $\Lambda$ for the observable $A = B_{q} \otimes I_{L-q}$ with $B_{q} \coloneqq T_{q} + T_{q}^{-1} - 2^{2-q} I_{q}$.
			Grey curves between colored curves show intermediate values of $q$, i.e., $q=4,8,12$.
			(b) The same data as the panel~(a) for $q=2,6,10$ and $14$ plotted against $L-q$.
			When $L\simeq q$ so that $L-q < L/2$, $\Lambda$ decreases as $\protect\propto 2^{-L}$.
			(c) The same data as the panel (a) for $q=6$.
			It decreases as $\propto 1.8^{-L}$ when $L-q \ll L/2$.
			When $L/2 \gtrsim L-q$, the decrease in $\Lambda$ becomes slower and follows a different exponential decay with a base of $1.8^{-1/2}$, instead of $1.8^{-1}$.
			Aside from the value of the base,
			this behavior is consistent with \eqref{eq_MainResult1}, which predicts that the exponent of the exponential decrease of $\Lambda$ in $L$ should be twice as large in the region $L-q \ll L/2$ compared to the region $L/2 \gtrsim L-q$.
            The standard errors are smaller than the size of the data points.
            The number of samples lies between 1000 and 10000 for each datum.
		}
	\end{figure}

\backmatter

% \bmhead{Supplementary information}
% If your article has accompanying supplementary file/s please state so here.
% Authors reporting data from electrophoretic gels and blots should supply the full unprocessed scans for key as part of their Supplementary information. This may be requested by the editorial team/s if it is missing.
% Please refer to Journal-level guidance for any specific requirements.

\bmhead{Acknowledgments}
LE, JH, and VR were supported by ERC Advanced Grant ``RMTBeyond'' No. 101020331.
SS was supported by KAKENHI Grant Number JP22J14935 from the Japan Society for the Promotion of Science (JSPS) and Forefront Physics and Mathematics Program to Drive Transformation (FoPM), a World-leading Innovative Graduate Study (WINGS) Program, the University of Tokyo.

\section*{Declarations}
\bmhead{Competing Interests}
The authors declare that there is no conflict of interest

% \begin{itemize}
% \item Conflict of interest/Competing interests (check journal-specific guidelines for which heading to use)
% \item Ethics approval
% \item Consent to participate
% \item Consent for publication
% \item Availability of data and materials
% \item Authors' contributions
% \end{itemize}

% \noindent
% If any of the sections are not relevant to your manuscript, please include the heading and write `Not applicable' for that section.

%%%%%%%%%%%%%%%%%%%%%%%%%%%%%%%%%%%%%%%%%%%%%%%%%%
%%%%%%%%%%%%%%%%%%%%%%%%%%%%%%%%%%%%%%%%%%%%%%%%%%

\begin{appendices}
	\section{Extension to higher dimensions}
	In this appendix, we extend our main result, Theorem~\ref{thm_mainTheoerem}, to the $d$-dimensional case.

	\subsection{Multidimensional setup}
		Let $\vb{L} \coloneqq (L_{1}, \ldots, L_{d})$ be a vector of positive integers and set $V \coloneqq \prod_{s=1}^{d} L_{s}$.
		We consider a $d$-dimensional system with $V$ quantum spins at the vertices of the classical discrete torus $$\mathbb{T}_{\vb{L}} \coloneqq \bigtimes_{s=1}^{d} \mathbb{Z}/L_{s} \mathbb{Z}\,.$$
		As before, on each vertex, the one particle Hilbert space is given by $\mathbb{C}^2$ with canonical basis $\{\ket{\uparrow}, \ket{\uparrow}\}$. The corresponding $V$-particle Hilbert space is given by
		$$\mathcal{H} \coloneqq \bigotimes_{s=1}^{V} \mathbb{C}^{2} \quad \text{with dimension} \quad \mbox{dim} \mathcal{H} = 2^V\,.
		$$
		%is simply given by the tensor product and hence has dimension  $2^{V}$.
		For a vector $\vb{q} = (q_{1},\ldots,q_{d}) \in \mathbb{T}_{\vb{L}}$, we introduce a rectangular subregion $\mathcal{R}_{\vb{q}} \subset \mathbb{T}_{\vb{L}}$ by
		\begin{equation*}
			\mathcal{R}_{\vb{q}}
			\coloneqq \Bqty\Big{ \vb{x} = (x_{1},\ldots,x_{d}) \in \mathbb{T}_{\vb{L}} \colon 1\leq x_{s} \leq q_{s}\,, \ s=1,\ldots,d}\,.
		\end{equation*}
		A self-adjoint operator of the form $A = A_{\vb{q}} \otimes I_{ \mathbb{T}_{\vb{L}} \setminus \mathcal{R}_{\vb{q}} }$ is called a  \textit{$\vb{q}$-local observable}, where $A_{\vb{q}}$ is self-adjoint and acts on the Hilbert space of the spins in $\mathcal{R}_{\vb{q}}$, and $I_{ \mathbb{T}_{\vb{L}}\setminus\mathcal{R}_{\vb{q}} }$ is the identity on $\mathbb{T}_{\vb{L}} \setminus \mathcal{R}_{\vb{q}}$.

		Finally, let $T_{s}$ be the (left) translation operator along the $s$-th coordinate acting on $\mathbb{T}_{\vb{L}}$.
		For a vector $\vb{j} \coloneqq (j_{1},\ldots, j_{d}) \in \mathbb{T}_{\vb{L}}$, we introduce $T^{ \vb{j} }
		\coloneqq \prod_{s=1}^{d} T_{s}^{j_{s}}$.

		The $d$-dimensional version of our model in Definition \ref{def:model} is given as follows.
		\begin{Def}
			Set the vector $\Bell \coloneqq (\ell_{1},\ldots, \ell_{d}) \in \mathbb{T}_{ \vb{L} }$ that determines the interaction range in each coordinate direction.
			We define the ensemble of Hamiltonians with local interactions as
			\begin{equation}
				H_{\vb{L}}^{({\Bell})}
				\coloneqq \sum_{ \vb{j} \in \mathbb{T}_{\vb{L}} } T^{-\vb{j}} \qty( h_{ \Bell } \otimes I_{ \mathbb{T}_{\vb{L}}\setminus
					\mathcal{R}_{\Bell} }  ) T^{\vb{j}}\, \quad \text{with} \quad   h_{\Bell} \coloneqq \sum_{p_{\vb{1}}, \ldots ,p_{\Bell} = 0}^{3} J_{p_{\vb{1}},\ldots, p_{\Bell}} \sigma_{\vb{1}}^{ (p_{\vb{1}}) } \ldots \sigma_{\Bell}^{ (p_{\Bell}) }
				\label{def_Hamiltonianapp}
			\end{equation}
			where the symbols $\vb{1}$, $\vb{2}$, \ldots, $\Bell$ label the elements of $\mathcal{R}_{\Bell}$ in an arbitrary order.
			As in \eqref{def_Hamiltonian}, $\sigma^{(p)}$ for $p \in \{0,1,2,3\}$ are the Pauli matrices \eqref{eq:Pauli}.

			The $4^{ \abs*{ \mathcal{R}_{\Bell} } }$ coefficients $J_{p_{\vb{1}},\ldots, p_{\Bell}} $ are i.i.d.~real Gaussian random variables with zero mean, $\mathbb{E} J_{p_{\vb{1}},\ldots, p_{\Bell}} =0$, and variance $v^2_{\Bell}\coloneqq \mathbb{E} \abs*{ J_{p_{\vb{1}},\ldots, p_{\Bell}} }^2 $.
		\end{Def}

	We have the following multidimensional analog of Lemma \ref{lemma_BlockDiagonalization},
	\begin{lemma}
		\label{lemma_BlockDiagonalization_HighDimension}
		Let $$\Pi_{\vb{k}} \coloneqq \frac{1}{V} \sum_{ \vb{j} \in \mathbb{T}_{\vb{L}} } e^{ 2\pi i \sum_{s=1}^{d} \frac{k_{s} j_{s}}{L_{s}} } T^{-\vb{j}} \quad \text{for} \quad  \vb{k} \in \mathbb{T}_{\vb{L}} $$
		be the projection operator onto the $\vb{k}$-momentum space, i.e.~$T_{s} \Pi_{ \vb{k} } = e^{ 2\pi i \frac{k_{s}}{L_{s}} } \Pi_{\vb{k}}$ for all $s=1,\ldots, d$.
		Then we have
		\begin{equation*}
			H_{\vb{L}}^{( \Bell )} = V \sum_{ \vb{k} \in \mathbb{T}_{\vb{L}} } \Pi_{\vb{k}} \qty( h_{ \Bell } \otimes I_{ \mathbb{T}_{\vb{L}}\setminus\mathcal{R}_{\Bell} }  ) \Pi_{\vb{k}}\,.
			\label{eq_ExpressionWithProjection_HighDimension}
		\end{equation*}
	\end{lemma}
	\begin{proof}
		This follows by Lemma~\ref{lemma_BlockDiagonalization} coordinatewise.
	\end{proof}

	Denoting by $\ket*{ E_{\alpha}^{(\vb{k})} }$ the normalised eigenvector of $H_{\vb{L}}^{(\Bell)}$ belonging to an eigenvalue $E_{\alpha}$ and the $\vb{k}$-momentum sector, i.e.~$H_{\vb{L}}^{(\vb{\ell})}\ket*{ E_{\alpha}^{(\vb{k})} } = E_\alpha \ket*{ E_{\alpha}^{(\vb{k})} }$ and $\Pi_{\vb{k}} \ket*{ E_{\alpha}^{(\vb{k})} } = \ket*{ E_{\alpha}^{(\vb{k})} }$, the definition of the \textit{microcanocical average} is completely analogous to Definition \ref{def:mcensemble}.

	Moreover, whenever we use the notation $\prec$ for stochastic domination (Definition~\ref{def_StochasticDomination}), it is always understood with $N\coloneqq  2^V$.

	\subsection{Multidimensional version of the main result}
	The $d$-dimensional version of Theorem~\ref{thm_mainTheoerem} is then given as follows.
	\begin{theorem}[ETH in $d$-dimensional translation-invariant systems]
		\label{thm_mainTheoerem_Ddimension_HighDimension}
		Let $\Bell = \vb{L}$ and consider the the Hamiltonian $H_{\vb{L}}^{(\vb{L})}$ from \eqref{def_Hamiltonianapp} with eigenvalues $E_\alpha^{(\vb{k})}$ and normalised eigenvectors $\ket*{E_\alpha^{(\vb{k})}}$. Then, for every $\Delta > 0$ and bounded $\vb{q}$-local observable $A = A_{ \vb{q} } \otimes I_{ \mathbb{T}_{ \vb{L} } \setminus \mathcal{R}_{\vb{q}} }$ with $q_{s} \leq L_{s}/2$ for all $s=1,\ldots, d$, it holds that
		\begin{equation}
			\max_{\alpha, \beta} \max_{\vb{k},\vb{k'}}    \abs{
				\matrixel*{ E_{\alpha}^{(\vb{k})} }{ A }{ E_{\beta}^{(\vb{k'})} } - \delta_{\alpha\beta} \delta_{\vb{k},\vb{k'}} \expval*{ A }^{(\mathrm{mc})}_{\Delta}(E_{\alpha}^{(\vb{k})})
			}
			\prec
			\frac{1}{ 2^{V/2} }\,.
			\label{eq_MainResult_HighDimension}
		\end{equation}
		That is, the ETH holds with optimal speed of convergence.
	\end{theorem}

	The principal strategy for proving Theorem \ref{thm_mainTheoerem_Ddimension_HighDimension} is exactly the same as for Theorem~\ref{thm_mainTheoerem}, which has been outlined right below Proposition \ref{prop:main}. We shall hence only discuss the differences compared to the proof in Section \ref{sec:result}, which consist solely of Step 1 (generalizing Lemma \ref{lemma_DimensionOfMomentumSectors}, cf.~Lemma \ref{lemma_DimensionOfMomentumSectors_HighDimension}) and Step 4 (generalizing Lemma~\ref{lemma_DifferenceBetweenSectors}, cf.~Lemma \ref{lemma_DimensionOfMomentumSectors_HighDimension}).

	\begin{lemma}[Step 1: Dimensions of momentum sectors]
		\label{lemma_DimensionOfMomentumSectors_HighDimension}
		The dimension $\tr_{ \vb{L} } \Pi_{\vb{k}}$ of the $\vb{k}$-momentum sectors for $\vb{k} \in \mathbb{T}_{ \vb{L} }$ is almost equal to each other in the sense that we have
		\begin{equation*}
			\tr_{ \vb{L} } \Pi_{\vb{k}} = \frac{2^{V}}{V} + \order{ 2^{V/2+(\log_2V)^2} }\,.
		\end{equation*}
	\end{lemma}
	\begin{proof}
		Let $\s=\s(\vb{L})$ denote the canonical product basis of $\mathcal{H}$, as in \eqref{eq:purestates}, and let $\G$ be the commutative group generated by the translation operators $\{T_s\}_{s=1}^d$. The action of $\G$ on $\s$ is defined by \eqref{eq:gactsons}.

		In general, the group $\G$ is not cyclic, hence the subgroups of $\G$ are not uniquely determined by their size. However, $\s$ can be decomposed into a disjoint union of sets $\s_{\mathcal{K}}=\s_{\mathcal{K}}(\vb{L})$ defined by
		\begin{equation*} \label{eq:s_K}
			\s_{\mathcal{K}} \coloneqq \{\sigma \in \s \,:\, \G_\sigma = \mathcal{K} \}\,,
		\end{equation*}
		where $\G_\sigma \subset \G$ is the stabilizer of $\sigma$ under the action \eqref{eq:gactsons}, and $\mathcal{K} \le \G$ is a subgroup of $\G$. Similarly to \eqref{eq:reduced_sigma}, for any subgroup $\mathcal{K}$ of $\G$, we define the map
		\begin{equation*}
			\varphi_{\mathcal{K}} : \s_{\mathcal{K}} \to \left(\faktor{\mathbb{T}_{\vb{L}} }{\mathcal{K}} \to \{\uparrow,\downarrow\}\right)\,, \quad  \bigl(\varphi_\mathcal{K}(\sigma)\bigr)([x]) \coloneqq \sigma(x)\,,  \quad [x] \in \faktor{\mathbb{T}_{\vb{L}}}{\mathcal{K}}\,,
		\end{equation*}
which is easily seen to be an injection and hence
		and injection,
		 $\abs*{ \s_\mathcal{K} }  \le 2^{V/ \abs*{ \mathcal{K} } }$.
Therefore, denoting the number of elements in $\s$ with a trivial stabilizer by $M(\vb{L})$, we obtain
		\begin{equation*}
			2^V = \sum_{\mathcal{K} \le G} \abs*{ \s_\mathcal{K} }  = M(\vb{L}) + \sum_{K \le G, \abs*{ \mathcal{K} } \ge 2} \abs*{ \s_\mathcal{K} }  \le M(\vb{L}) + s(\G)2^{V/2}\,,
		\end{equation*}
		where $s(\G)$ denotes the number of subgroups of $\G$.
		Combining this with the following well-known bound\footnote{More precisely, in order to see that \eqref{eq:num_subgroups} holds, observe that for any subgroup $\mathcal{K}$ of $\G$ and any $g \in \G\backslash\mathcal{K}$, the size of the subgroup generated by $\mathcal{K}$ and $g$ is at least $2 \abs*{ \mathcal{K} } $. Therefore, any subgroup $\mathcal{K}$ is generated by at most $\log_2 \abs*{ \G } $ elements, hence the set of all subgroups of $\G$ can be injectively mapped to $\G^{\log_2 \abs*{ \G } }$.}
		\begin{equation} \label{eq:num_subgroups}
			s(\G) \le \abs*{ \G } ^{\log_2 \abs*{ \G } }
		\end{equation}
	and the trivial estimate $M(\vb{L}) \le 2^V$, we conclude that
		\begin{equation} \label{eq:multi_d_bound}
			M(\vb{L}) = 2^V + \order{2^{V/2+(\log_2V)^2}}.
		\end{equation}

		The construction of a linearly independent vectors with a fixed momentum $\vb{k} \in \mathbb{T}_{\vb{L}}$ for each disjoint orbit with a trivial stabilizer is analogous to \eqref{eq:vsigmadef}. Using estimates analogous to \eqref{eq:tr_lower_bound} and \eqref{eq:tr_upper_bound} together with \eqref{eq:multi_d_bound} concludes the proof of Lemma~\ref{lemma_DimensionOfMomentumSectors_HighDimension}.
	\end{proof}

	Finally, we discuss the generalisation of Step 4, i.e.~Lemma \ref{lemma_DifferenceBetweenSectors}.

	\begin{lemma}[Step 4: Traces within momentum sectors]
		\label{lemma_DifferenceBetweenSectors_HighDimension}
		Let $A = A_{ \vb{q} } \otimes I_{ \mathbb{T}_{ \vb{L} } \setminus \mathcal{R}_{\vb{q}} }$ be a bounded $\vb{q}$-local observable with $q_{s} \leq L_{s}/2$ for all $s=1,\ldots, d$.
		Then it holds that
		\begin{equation}
			\max_{\vb{k}}    \abs{ \frac{ \tr_{ \vb{L} }\qty( \Pi_{ \vb{k} } A \Pi_{ \vb{k} } ) }{ \tr_{ \vb{L} } \Pi_{ \vb{k} } } - \expval*{ A } }
			\leq \order{ \frac{ V }{ 2^{V/2} } }\,.
			\label{eq_DifferenceBetweenSectors_HighDimension}
		\end{equation}
	\end{lemma}
	\begin{proof}
		Substituting $\Pi_{\vb{k}} \coloneqq \frac{1}{V} \sum_{ \vb{j} \in \mathbb{T}_{\vb{L}} } e^{ 2\pi i \sum_{s=1}^{d} \frac{k_{s} j_{s}}{L_{s}} } T^{-\vb{j}} $ for $ \vb{k} \in \mathbb{T}_{\vb{L}}$, we obtain
		\begin{align*}
			\tr_{ \vb{L} }\qty( \Pi_{ \vb{k} } A \Pi_{ \vb{k} } )
			&= \expval*{ A } \tr_{ \vb{L} } \Pi_{ \vb{k} } + \frac{1}{V} \sum_{ \vb{j} \in \mathbb{T}_{\vb{L}} \setminus \Bqty*{\vb{0}} } e^{ 2\pi i \sum_{s=1}^{d} \frac{k_{s} j_{s}}{L_{s}} } \tr_{ \vb{L} }( T^{-\vb{j}} \mathring{A} )\,.
			\label{TraceWithinMomentumSector_High}
		\end{align*}
		Then, the task is to evaluate the size of the quantity $\tr_{ \vb{L} }( T^{-\vb{j}} \mathring{A} )$.
		\begin{lemma}
			\label{lemma_BoundForTraceWithTranslation_HighDimension}
			Let $A \coloneqq A_{ \vb{q} } \otimes I_{ \mathbb{T}_{ \vb{L} } \setminus \mathcal{R}_{\vb{q}} }$ be a $\vb{q}$-local observable with $q_{s} \leq L_{s} / 2$ for all $s=1,\ldots,d$ and $\Vert A \Vert \lesssim 1$.
			Then, for any $\vb{j} \in \mathbb{T}_{ \vb{L} } \setminus \Bqty*{ \vb{0} }$, we have that
			\begin{equation}
				\abs{ \tr_{ \vb{L} }( T^{-\vb{j}} A ) }
				\lesssim 2^{V/2 }\,.
				\label{eq_BoundForTraceWithTranslation_HighDimension}
			\end{equation}
		\end{lemma}
		Combining \eqref{eq_BoundForTraceWithTranslation_HighDimension} with Lemma~\ref{lemma_DimensionOfMomentumSectors_HighDimension} gives the bound~\eqref{eq_DifferenceBetweenSectors_HighDimension}.
	\end{proof}

	Its remains to prove Lemma \ref{lemma_BoundForTraceWithTranslation_HighDimension}.

	\begin{proof}[Proof of Lemma \ref{lemma_BoundForTraceWithTranslation_HighDimension}]
		The $d=1$ case is proven in Lemma~\ref{lemma_BoundForTraceWithTranslation}.
		Thus, we assume $d\geq 2$ in the following.
		We choose an orthonormal basis of the Hilbert space on $\mathbb{T}_{ \vb{L} }$ as $\Bqty*{ \ket*{s} \mid s \colon \mathbb{T}_{ \vb{L} } \to \Bqty*{\up, \dn} }$ to calculate the trace.
		Then, similarly to \eqref{j_term}, we obtain
		\begin{align}
			\abs{ \tr_{ \vb{L} }( T^{-\vb{j}} A ) }
			\lesssim \sum_{s \colon \mathbb{T}_{ \vb{L} } \to \Bqty*{\up,\dn} } \prod_{ \vb{x} \in \mathbb{T}_{ \vb{L} } \setminus \mathcal{R}_{\vb{q}} } \delta_{ s(\vb{x}), s(\vb{x} + \vb{j}) }\,.
			\label{eq_TraceWithTranslatin_HighDimension1}
		\end{align}

		Next, analogously to \eqref{eq:Gdef}, we count the number of independent summations on the right-hand side of \eqref{eq_TraceWithTranslatin_HighDimension1}.
		To do so, we consider a graph $\mathcal{G}_{\vb{L},\vb{q},\vb{j}} = (V,E)$, whose vertices and edges are given by $V \coloneqq \mathbb{T}_{ \vb{L} }$ and $E \coloneqq \Bqty*{ (\vb{x}, \vb{x}+\vb{j}) \colon \vb{x} \in \mathbb{T}_{ \vb{L} }\setminus \mathcal{R}_{\vb{q}} }$, respectively.
		Exactly one redundant delta function appears in the product $\prod_{ \vb{x} \in \mathbb{T}_{ \vb{L} } \setminus \mathcal{R}_{\vb{q}} } \delta_{ s(\vb{x}), s(\vb{x} + \vb{j}) }$ for every occurrence of a loop in $\mathcal{G}_{\vb{L},\vb{q},\vb{j}}$.
		Thus, by denoting the number of loops in $\mathcal{G}_{\vb{L},\vb{q},\vb{j}}$ by $N(\vb{L}, \vb{q}, \vb{j})$, we obtain
		\begin{align}	\label{eq_TraceWithTranslatin_HighDimension2}
			\abs{ \tr_{ \vb{L} }( T^{-\vb{j}} \mathring{A} ) }
			\lesssim 2^{ \abs*{
					\mathcal{R}_{\vb{q}} } + N(\vb{L}, \vb{q}, \vb{j}) }\,.
		\end{align}

		As in the one-dimensional case, Lemma \ref{lemma_BoundForTraceWithTranslation}, the graph $\mathcal{G}_{\vb{L},\vb{q},\vb{j}}$ is obtained from $\mathcal{G}_{\vb{L},\vb{0},\vb{j}}$ by removing the edge $(\vb{x}, \vb{x}+\vb{j})$ for all $\vb{x} \in \mathcal{R}_{\vb{q}}$.
		Therefore, we have $N(\vb{L}, \vb{q}, \vb{j}) \leq N(\vb{L}, \vb{0}, \vb{j})$ for all~$\vb{q}$. Now, the number of loops in $\mathcal{G}_{\vb{L},\vb{0},\vb{j}}$ can be counted by considering the orbits of the cyclic group $\expval*{ T^{\vb{j}} }$ on $\mathbb{T}_{ \vb{L} }$.
		It is clear that the size of each orbit is equal to one another.
		Denoting it by $g(\vb{j})$, the number of loops in $\mathcal{G}_{\vb{L},\vb{0},\vb{j}}$ is given by
		\begin{equation*}
			N(\vb{L}, \vb{0}, \vb{j}) = \frac{ V }{ g(\vb{j}) }\,.
		\end{equation*}
		Note that, since $\vb{j} \neq \vb{0}$ by assumption, we have $g(\vb{j}) \geq 2$.

		If $g(\vb{j}) \geq 4$, the bound \eqref{eq_BoundForTraceWithTranslation_HighDimension} is already proven because
		\begin{align*}
			\abs{ \tr_{ \vb{L} }( T^{-\vb{j}} \mathring{A} ) }
			\lesssim 2^{ \abs*{
					\mathcal{R}_{\vb{q}} } + N(\vb{L}, \vb{q}, \vb{j}) }
			\leq 2^{ \frac{V}{4} + \frac{V}{4} } = 2^{ \frac{V}{2} }\,,
			\label{eq_TraceWithTranslatin_HighDimension3}
		\end{align*}
		where we used $\abs*{ \mathcal{R}_{\vb{q}} } \leq V/2^{d}$ and $d\geq 2$ by assumption.

		If $g(\vb{j}) = 2$ or $g(\vb{j}) = 3$, we must have $g(\vb{j}) j_{s} \equiv 0\, \pmod{L_{s}}$ for all $s$.
		Using again that $\vb{j} \neq \vb{0}$, there exists a non-zero component $j_{t}$.
		For such a coordinate direction $t \in \{1, ... , d\}$, we must have $g(\vb{j}) \mid L_{t} $ because $g(\vb{j})\in\{2,3\}$ is prime. We hence have a decomposition
		\begin{equation*}
			\mathbb{T}_{ \vb{L} }
			= \mathcal{A}_{t} \sqcup T^{\vb{j}} \mathcal{A}_{t} \sqcup T^{2 \vb{j}} \mathcal{A}_{t}\qc
			\mathcal{A}_{t} \coloneqq \Bqty{ \vb{x} \in \mathbb{T}_{ \vb{L} } \colon 1 \leq x_{t} \leq \frac{ L_{t} }{ g(\vb{j}) } }\,.
		\end{equation*}

		Every loop in $\mathcal{G}_{\vb{L},\vb{0},\vb{j}}$ can be considered to start from a site in $\mathcal{A}_{t}$.
		Therefore, removing the edge $(\vb{x}, \vb{x}+\vb{j})$ for all $\vb{x} \in \mathcal{R}_{\vb{q}}$ from $\mathcal{G}_{\vb{L},\vb{0},\vb{j}}$ decreases the number of loops at least by $\abs{ \mathcal{R}_{ \vb{q} } \cap \mathcal{A}_{t} }$, which implies
		\begin{equation*}
			N(\vb{L}, \vb{q}, \vb{j}) \leq N(\vb{L}, \vb{0}, \vb{j}) - \abs{ \mathcal{R}_{ \vb{q} } \cap \mathcal{A}_{t} }.
		\end{equation*}
		Thus, from \eqref{eq_TraceWithTranslatin_HighDimension2}, we obtain
		\begin{align*}
			\abs{ \tr_{ \vb{L} }( T^{-\vb{j}} \mathring{A} ) }
			\lesssim 2^{ \abs*{
					\mathcal{R}_{\vb{q}} } - \abs{ \mathcal{R}_{ \vb{q} } \cap \mathcal{A}_{t} } + N(\vb{L}, \vb{0}, \vb{j}) }
			= 2^{ \abs{ \mathcal{R}_{ \vb{q} } \setminus \mathcal{A}_{t} } + N(\vb{L}, \vb{0}, \vb{j}) }\,.
		\end{align*}
		Finally, we have
		\begin{align*}
			\abs{ \mathcal{R}_{ \vb{q} } \setminus \mathcal{A}_{t} } +
			N(\vb{L}, \vb{0}, \vb{j})
			\le \qty( \frac{ L_{t} }{2} - \frac{L_{t}}{ g(\vb{j}) } ) \prod_{s (\neq t) } \frac{ L_{s} }{2}
			+ \frac{ V }{ g(\vb{j}) } = \frac{ V }{ 2^{d} } \qty( 1- \frac{2}{ g(\vb{j}) } ) + \frac{ V }{ g(\vb{j}) } \leq \frac{V}{2}\,,
		\end{align*}
		which completes the proof of Lemma~\ref{lemma_BoundForTraceWithTranslation_HighDimension}.
	\end{proof}
\end{appendices}


\bibliography{refs}% common bib file


\begin{thebibliography}{35}
\ifx \bisbn   \undefined \def \bisbn  #1{ISBN #1}\fi
\ifx \binits  \undefined \def \binits#1{#1}\fi
\ifx \bauthor  \undefined \def \bauthor#1{#1}\fi
\ifx \batitle  \undefined \def \batitle#1{#1}\fi
\ifx \bjtitle  \undefined \def \bjtitle#1{#1}\fi
\ifx \bvolume  \undefined \def \bvolume#1{\textbf{#1}}\fi
\ifx \byear  \undefined \def \byear#1{#1}\fi
\ifx \bissue  \undefined \def \bissue#1{#1}\fi
\ifx \bfpage  \undefined \def \bfpage#1{#1}\fi
\ifx \blpage  \undefined \def \blpage #1{#1}\fi
\ifx \burl  \undefined \def \burl#1{\textsf{#1}}\fi
\ifx \doiurl  \undefined \def \doiurl#1{\url{https://doi.org/#1}}\fi
\ifx \betal  \undefined \def \betal{\textit{et al.}}\fi
\ifx \binstitute  \undefined \def \binstitute#1{#1}\fi
\ifx \binstitutionaled  \undefined \def \binstitutionaled#1{#1}\fi
\ifx \bctitle  \undefined \def \bctitle#1{#1}\fi
\ifx \beditor  \undefined \def \beditor#1{#1}\fi
\ifx \bpublisher  \undefined \def \bpublisher#1{#1}\fi
\ifx \bbtitle  \undefined \def \bbtitle#1{#1}\fi
\ifx \bedition  \undefined \def \bedition#1{#1}\fi
\ifx \bseriesno  \undefined \def \bseriesno#1{#1}\fi
\ifx \blocation  \undefined \def \blocation#1{#1}\fi
\ifx \bsertitle  \undefined \def \bsertitle#1{#1}\fi
\ifx \bsnm \undefined \def \bsnm#1{#1}\fi
\ifx \bsuffix \undefined \def \bsuffix#1{#1}\fi
\ifx \bparticle \undefined \def \bparticle#1{#1}\fi
\ifx \barticle \undefined \def \barticle#1{#1}\fi
\bibcommenthead
\ifx \bconfdate \undefined \def \bconfdate #1{#1}\fi
\ifx \botherref \undefined \def \botherref #1{#1}\fi
\ifx \url \undefined \def \url#1{\textsf{#1}}\fi
\ifx \bchapter \undefined \def \bchapter#1{#1}\fi
\ifx \bbook \undefined \def \bbook#1{#1}\fi
\ifx \bcomment \undefined \def \bcomment#1{#1}\fi
\ifx \oauthor \undefined \def \oauthor#1{#1}\fi
\ifx \citeauthoryear \undefined \def \citeauthoryear#1{#1}\fi
\ifx \endbibitem  \undefined \def \endbibitem {}\fi
\ifx \bconflocation  \undefined \def \bconflocation#1{#1}\fi
\ifx \arxivurl  \undefined \def \arxivurl#1{\textsf{#1}}\fi
\csname PreBibitemsHook\endcsname

%%% 1
\bibitem[\protect\citeauthoryear{Santos and
  Rigol}{2010}]{santos2010localization}
\begin{barticle}
\bauthor{\bsnm{Santos}, \binits{L.F.}},
\bauthor{\bsnm{Rigol}, \binits{M.}}:
\batitle{{Localization and the effects of symmetries in the thermalization
  properties of one-dimensional quantum systems}}.
\bjtitle{Physical Review E}
\bvolume{82}(\bissue{3}),
\bfpage{031130}
(\byear{2010})
\doiurl{10.1103/PhysRevE.82.031130}
\end{barticle}
\endbibitem

%%% 2
\bibitem[\protect\citeauthoryear{Trotzky et~al.}{2012}]{trotzky2012probing}
\begin{barticle}
\bauthor{\bsnm{Trotzky}, \binits{S.}},
\bauthor{\bsnm{Chen}, \binits{A.} \bsuffix{Yu{-}Ao}},
\bauthor{\bsnm{Flesch}, \binits{A.}},
\bauthor{\bsnm{McCulloch}, \binits{I.P.}},
\bauthor{\bsnm{Schollw{\"{o}}ck}, \binits{U.}},
\bauthor{\bsnm{Eisert}, \binits{J.}},
\bauthor{\bsnm{Bloch}, \binits{I.}}:
\batitle{{Probing the relaxation towards equilibrium in an isolated strongly
  correlated one-dimensional Bose gas}}.
\bjtitle{Nature Physics}
\bvolume{8}(\bissue{4}),
\bfpage{325}--\blpage{330}
(\byear{2012})
\doiurl{10.1038/nphys2232}
\end{barticle}
\endbibitem

%%% 3
\bibitem[\protect\citeauthoryear{Langen et~al.}{2013}]{langen2013local}
\begin{barticle}
\bauthor{\bsnm{Langen}, \binits{T.}},
\bauthor{\bsnm{Geiger}, \binits{R.}},
\bauthor{\bsnm{Kuhnert}, \binits{M.}},
\bauthor{\bsnm{Rauer}, \binits{B.}},
\bauthor{\bsnm{Schmiedmayer}, \binits{J.}}:
\batitle{{Local emergence of thermal correlations in an isolated quantum
  many-body system}}.
\bjtitle{Nature Physics}
\bvolume{9}(\bissue{10}),
\bfpage{640}--\blpage{643}
(\byear{2013})
\end{barticle}
\endbibitem

%%% 4
\bibitem[\protect\citeauthoryear{Clos et~al.}{2016}]{clos2016time}
\begin{barticle}
\bauthor{\bsnm{Clos}, \binits{G.}},
\bauthor{\bsnm{Porras}, \binits{D.}},
\bauthor{\bsnm{Warring}, \binits{U.}},
\bauthor{\bsnm{Schaetz}, \binits{T.}}:
\batitle{{Time-Resolved Observation of Thermalization in an Isolated Quantum
  System}}.
\bjtitle{Physical Review Letters}
\bvolume{117}(\bissue{17}),
\bfpage{170401}
(\byear{2016})
\doiurl{10.1103/PhysRevLett.117.170401}
\end{barticle}
\endbibitem

%%% 5
\bibitem[\protect\citeauthoryear{Kaufman et~al.}{2016}]{kaufman2016quantum}
\begin{barticle}
\bauthor{\bsnm{Kaufman}, \binits{A.M.}},
\bauthor{\bsnm{Tai}, \binits{M.E.}},
\bauthor{\bsnm{Lukin}, \binits{A.}},
\bauthor{\bsnm{Rispoli}, \binits{M.}},
\bauthor{\bsnm{Schittko}, \binits{R.}},
\bauthor{\bsnm{Preiss}, \binits{P.M.}},
\bauthor{\bsnm{Greiner}, \binits{M.}}:
\batitle{{Quantum thermalization through entanglement in an isolated many-body
  system}}.
\bjtitle{Science}
\bvolume{353}(\bissue{6301}),
\bfpage{794}--\blpage{800}
(\byear{2016})
\doiurl{10.1126/science.aaf6725}
{\href{https://arxiv.org/abs/1603.04409}{{1603.04409}}}
\end{barticle}
\endbibitem

%%% 6
\bibitem[\protect\citeauthoryear{Neill et~al.}{2016}]{neill2016ergodic}
\begin{barticle}
\bauthor{\bsnm{Neill}, \binits{C.}},
\bauthor{\bsnm{Roushan}, \binits{P.}},
\bauthor{\bsnm{Fang}, \binits{M.}},
\bauthor{\bsnm{Chen}, \binits{Y.}},
\bauthor{\bsnm{Kolodrubetz}, \binits{M.}},
\bauthor{\bsnm{Chen}, \binits{Z.}},
\bauthor{\bsnm{Megrant}, \binits{A.}},
\bauthor{\bsnm{Barends}, \binits{R.}},
\bauthor{\bsnm{Campbell}, \binits{B.}},
\bauthor{\bsnm{Chiaro}, \binits{B.}},
\bauthor{\bsnm{Dunsworth}, \binits{A.}},
\bauthor{\bsnm{Jeffrey}, \binits{E.}},
\bauthor{\bsnm{Kelly}, \binits{J.}},
\bauthor{\bsnm{Mutus}, \binits{J.}},
\bauthor{\bsnm{O'Malley}, \binits{P.J.J.}},
\bauthor{\bsnm{Quintana}, \binits{C.}},
\bauthor{\bsnm{Sank}, \binits{D.}},
\bauthor{\bsnm{Vainsencher}, \binits{A.}},
\bauthor{\bsnm{Wenner}, \binits{J.}},
\bauthor{\bsnm{White}, \binits{T.C.}},
\bauthor{\bsnm{Polkovnikov}, \binits{A.}},
\bauthor{\bsnm{Martinis}, \binits{J.M.}}:
\batitle{{Ergodic dynamics and thermalization in an isolated quantum system}}.
\bjtitle{Nature Physics}
\bvolume{12}(\bissue{11}),
\bfpage{1037}--\blpage{1041}
(\byear{2016})
\doiurl{10.1038/nphys3830}
\end{barticle}
\endbibitem

%%% 7
\bibitem[\protect\citeauthoryear{Tang et~al.}{2018}]{tang2018thermalization}
\begin{barticle}
\bauthor{\bsnm{Tang}, \binits{Y.}},
\bauthor{\bsnm{Kao}, \binits{W.}},
\bauthor{\bsnm{Li}, \binits{K.-Y.}},
\bauthor{\bsnm{Seo}, \binits{S.}},
\bauthor{\bsnm{Mallayya}, \binits{K.}},
\bauthor{\bsnm{Rigol}, \binits{M.}},
\bauthor{\bsnm{Gopalakrishnan}, \binits{S.}},
\bauthor{\bsnm{Lev}, \binits{B.L.}}:
\batitle{{Thermalization near Integrability in a Dipolar Quantum Newton’s
  Cradle}}.
\bjtitle{Physical Review X}
\bvolume{8}(\bissue{2}),
\bfpage{021030}
(\byear{2018})
\doiurl{10.1103/PhysRevX.8.021030}
{\href{https://arxiv.org/abs/1707.07031}{{1707.07031}}}
\end{barticle}
\endbibitem

%%% 8
\bibitem[\protect\citeauthoryear{von Neumann}{2010}]{von2010proof}
\begin{barticle}
\bauthor{\bsnm{Neumann}, \binits{J.}}:
\batitle{{Proof of the ergodic theorem and the H-theorem in quantum
  mechanics}}.
\bjtitle{The European Physical Journal H}
\bvolume{35}(\bissue{2}),
\bfpage{201}--\blpage{237}
(\byear{2010})
\end{barticle}
\endbibitem

%%% 9
\bibitem[\protect\citeauthoryear{Deutsch}{1991}]{deutsch1991quantum}
\begin{barticle}
\bauthor{\bsnm{Deutsch}, \binits{J.M.}}:
\batitle{{Quantum statistical mechanics in a closed system}}.
\bjtitle{Physical Review A}
\bvolume{43}(\bissue{4}),
\bfpage{2046}
(\byear{1991})
\end{barticle}
\endbibitem

%%% 10
\bibitem[\protect\citeauthoryear{Srednicki}{1994}]{srednicki1994chaos}
\begin{barticle}
\bauthor{\bsnm{Srednicki}, \binits{M.}}:
\batitle{{Chaos and quantum thermalization}}.
\bjtitle{Physical Review E}
\bvolume{50}(\bissue{2}),
\bfpage{888}
(\byear{1994})
\end{barticle}
\endbibitem

%%% 11
\bibitem[\protect\citeauthoryear{D'Alessio et~al.}{2016}]{d2016quantum}
\begin{barticle}
\bauthor{\bsnm{D'Alessio}, \binits{L.}},
\bauthor{\bsnm{Kafri}, \binits{Y.}},
\bauthor{\bsnm{Polkovnikov}, \binits{A.}},
\bauthor{\bsnm{Rigol}, \binits{M.}}:
\batitle{{From quantum chaos and eigenstate thermalization to statistical
  mechanics and thermodynamics}}.
\bjtitle{Advances in Physics}
\bvolume{65}(\bissue{3}),
\bfpage{239}--\blpage{362}
(\byear{2016})
\doiurl{10.1080/00018732.2016.1198134}
{\href{https://arxiv.org/abs/1509.06411}{{1509.06411}}}
\end{barticle}
\endbibitem

%%% 12
\bibitem[\protect\citeauthoryear{Mori et~al.}{2018}]{mori2018thermalization}
\begin{barticle}
\bauthor{\bsnm{Mori}, \binits{T.}},
\bauthor{\bsnm{Ikeda}, \binits{T.N.}},
\bauthor{\bsnm{Kaminishi}, \binits{E.}},
\bauthor{\bsnm{Ueda}, \binits{M.}}:
\batitle{{Thermalization and prethermalization in isolated quantum systems: a
  theoretical overview}}.
\bjtitle{Journal of Physics B: Atomic, Molecular and Optical Physics}
\bvolume{51}(\bissue{11}),
\bfpage{112001}
(\byear{2018})
\doiurl{10.1088/1361-6455/aabcdf}
\end{barticle}
\endbibitem

%%% 13
\bibitem[\protect\citeauthoryear{Deutsch}{2018}]{deutsch2018eigenstate}
\begin{barticle}
\bauthor{\bsnm{Deutsch}, \binits{J.M.}}:
\batitle{{Eigenstate thermalization hypothesis}}.
\bjtitle{Reports on Progress in Physics}
\bvolume{81}(\bissue{8}),
\bfpage{082001}
(\byear{2018})
\doiurl{10.1088/1361-6633/aac9f1}
{\href{https://arxiv.org/abs/1805.01616}{{1805.01616}}}
\end{barticle}
\endbibitem

%%% 14
\bibitem[\protect\citeauthoryear{Rigol et~al.}{2008}]{rigol2008thermalization}
\begin{barticle}
\bauthor{\bsnm{Rigol}, \binits{M.}},
\bauthor{\bsnm{Dunjko}, \binits{V.}},
\bauthor{\bsnm{Olshanii}, \binits{M.}}:
\batitle{{Thermalization and its mechanism for generic isolated quantum
  systems}}.
\bjtitle{Nature}
\bvolume{452}(\bissue{7189}),
\bfpage{854}--\blpage{858}
(\byear{2008})
\doiurl{10.1038/nature06838}
\end{barticle}
\endbibitem

%%% 15
\bibitem[\protect\citeauthoryear{Rigol}{2009a}]{rigol2009breakdown}
\begin{barticle}
\bauthor{\bsnm{Rigol}, \binits{M.}}:
\batitle{{Breakdown of Thermalization in Finite One-Dimensional Systems}}.
\bjtitle{Physical Review Letters}
\bvolume{103}(\bissue{10}),
\bfpage{100403}
(\byear{2009})
\doiurl{10.1103/PhysRevLett.103.100403}
\end{barticle}
\endbibitem

%%% 16
\bibitem[\protect\citeauthoryear{Rigol}{2009b}]{rigol2009quantum}
\begin{barticle}
\bauthor{\bsnm{Rigol}, \binits{M.}}:
\batitle{{Quantum quenches and thermalization in one-dimensional fermionic
  systems}}.
\bjtitle{Physical Review A}
\bvolume{80}(\bissue{5}),
\bfpage{053607}
(\byear{2009})
\doiurl{10.1103/PhysRevA.80.053607}
{\href{https://arxiv.org/abs/0908.3188}{{0908.3188}}}
\end{barticle}
\endbibitem

%%% 17
\bibitem[\protect\citeauthoryear{Biroli et~al.}{2010}]{biroli2010effect}
\begin{barticle}
\bauthor{\bsnm{Biroli}, \binits{G.}},
\bauthor{\bsnm{Kollath}, \binits{C.}},
\bauthor{\bsnm{L{\"{a}}uchli}, \binits{A.M.}}:
\batitle{{Effect of Rare Fluctuations on the Thermalization of Isolated Quantum
  Systems}}.
\bjtitle{Physical Review Letters}
\bvolume{105}(\bissue{25}),
\bfpage{250401}
(\byear{2010})
\doiurl{10.1103/PhysRevLett.105.250401}
{\href{https://arxiv.org/abs/0907.3731}{{0907.3731}}}
\end{barticle}
\endbibitem

%%% 18
\bibitem[\protect\citeauthoryear{Steinigeweg
  et~al.}{2013}]{steinigeweg2013eigenstate}
\begin{barticle}
\bauthor{\bsnm{Steinigeweg}, \binits{R.}},
\bauthor{\bsnm{Herbrych}, \binits{J.}},
\bauthor{\bsnm{Prelov{\v{s}}ek}, \binits{P.}}:
\batitle{{Eigenstate thermalization within isolated spin-chain systems}}.
\bjtitle{Physical Review E}
\bvolume{87}(\bissue{1}),
\bfpage{012118}
(\byear{2013})
\doiurl{10.1103/PhysRevE.87.012118}
\end{barticle}
\endbibitem

%%% 19
\bibitem[\protect\citeauthoryear{Beugeling et~al.}{2014}]{beugeling2014finite}
\begin{barticle}
\bauthor{\bsnm{Beugeling}, \binits{W.}},
\bauthor{\bsnm{Moessner}, \binits{R.}},
\bauthor{\bsnm{Haque}, \binits{M.}}:
\batitle{{Finite-size scaling of eigenstate thermalization}}.
\bjtitle{Physical Review E}
\bvolume{89}(\bissue{4}),
\bfpage{042112}
(\byear{2014})
\doiurl{10.1103/PhysRevE.89.042112}
\end{barticle}
\endbibitem

%%% 20
\bibitem[\protect\citeauthoryear{Rigol et~al.}{2007}]{rigol2007relaxation}
\begin{barticle}
\bauthor{\bsnm{Rigol}, \binits{M.}},
\bauthor{\bsnm{Dunjko}, \binits{V.}},
\bauthor{\bsnm{Yurovsky}, \binits{V.}},
\bauthor{\bsnm{Olshanii}, \binits{M.}}:
\batitle{{Relaxation in a Completely Integrable Many-Body Quantum System: An Ab
  Initio Study of the Dynamics of the Highly Excited States of 1D Lattice
  Hard-Core Bosons}}.
\bjtitle{Physical Review Letters}
\bvolume{98}(\bissue{5}),
\bfpage{050405}
(\byear{2007})
\end{barticle}
\endbibitem

%%% 21
\bibitem[\protect\citeauthoryear{Cassidy et~al.}{2011}]{cassidy2011generalized}
\begin{barticle}
\bauthor{\bsnm{Cassidy}, \binits{A.C.}},
\bauthor{\bsnm{Clark}, \binits{C.W.}},
\bauthor{\bsnm{Rigol}, \binits{M.}}:
\batitle{{Generalized Thermalization in an Integrable Lattice System}}.
\bjtitle{Physical Review Letters}
\bvolume{106}(\bissue{14}),
\bfpage{140405}
(\byear{2011})
\end{barticle}
\endbibitem

%%% 22
\bibitem[\protect\citeauthoryear{Hamazaki
  et~al.}{2016}]{hamazaki2016generalized}
\begin{barticle}
\bauthor{\bsnm{Hamazaki}, \binits{R.}},
\bauthor{\bsnm{Ikeda}, \binits{T.N.}},
\bauthor{\bsnm{Ueda}, \binits{M.}}:
\batitle{{Generalized Gibbs ensemble in a nonintegrable system with an
  extensive number of local symmetries}}.
\bjtitle{Physical Review E}
\bvolume{93}(\bissue{3}),
\bfpage{032116}
(\byear{2016})
\end{barticle}
\endbibitem

%%% 23
\bibitem[\protect\citeauthoryear{Basko et~al.}{2006}]{basko2006metal}
\begin{barticle}
\bauthor{\bsnm{Basko}, \binits{D.M.}},
\bauthor{\bsnm{Aleiner}, \binits{I.L.}},
\bauthor{\bsnm{Altshuler}, \binits{B.L.}}:
\batitle{{Metal--insulator transition in a weakly interacting many-electron
  system with localized single-particle states}}.
\bjtitle{Annals of physics}
\bvolume{321}(\bissue{5}),
\bfpage{1126}--\blpage{1205}
(\byear{2006})
\end{barticle}
\endbibitem

%%% 24
\bibitem[\protect\citeauthoryear{Nandkishore and
  Huse}{2015}]{nandkishore2015many}
\begin{barticle}
\bauthor{\bsnm{Nandkishore}, \binits{R.}},
\bauthor{\bsnm{Huse}, \binits{D.A.}}:
\batitle{{Many-Body Localization and Thermalization in Quantum Statistical
  Mechanics}}.
\bjtitle{Annual Review of Condensed Matter Physics}
\bvolume{6}(\bissue{1}),
\bfpage{15}--\blpage{38}
(\byear{2015})
\end{barticle}
\endbibitem

%%% 25
\bibitem[\protect\citeauthoryear{Turner et~al.}{2018}]{turner2018weak}
\begin{barticle}
\bauthor{\bsnm{Turner}, \binits{C.J.}},
\bauthor{\bsnm{Michailidis}, \binits{A.A.}},
\bauthor{\bsnm{Abanin}, \binits{D.A.}},
\bauthor{\bsnm{Serbyn}, \binits{M.}},
\bauthor{\bsnm{Papi{\'c}}, \binits{Z.}}:
\batitle{{Weak ergodicity breaking from quantum many-body scars}}.
\bjtitle{Nature Physics}
\bvolume{14}(\bissue{7}),
\bfpage{745}--\blpage{749}
(\byear{2018})
\end{barticle}
\endbibitem

%%% 26
\bibitem[\protect\citeauthoryear{Bull et~al.}{2019}]{bull2019systematic}
\begin{barticle}
\bauthor{\bsnm{Bull}, \binits{K.}},
\bauthor{\bsnm{Martin}, \binits{I.}},
\bauthor{\bsnm{Papi{\'c}}, \binits{Z.}}:
\batitle{{Systematic Construction of Scarred Many-Body Dynamics in 1D Lattice
  Models}}.
\bjtitle{Physical Review Letters}
\bvolume{123}(\bissue{3}),
\bfpage{030601}
(\byear{2019})
\end{barticle}
\endbibitem

%%% 27
\bibitem[\protect\citeauthoryear{Reimann}{2015}]{reimann2015generalization}
\begin{barticle}
\bauthor{\bsnm{Reimann}, \binits{P.}}:
\batitle{{Generalization of von Neumann's Approach to Thermalization}}.
\bjtitle{Physical Review Letters}
\bvolume{115}(\bissue{1}),
\bfpage{010403}
(\byear{2015})
\doiurl{10.1103/PhysRevLett.115.010403}
{\href{https://arxiv.org/abs/1507.00262}{{1507.00262}}}
\end{barticle}
\endbibitem

%%% 28
\bibitem[\protect\citeauthoryear{Cipolloni
  et~al.}{2021}]{cipolloni2021eigenstate}
\begin{barticle}
\bauthor{\bsnm{Cipolloni}, \binits{G.}},
\bauthor{\bsnm{Erdős}, \binits{L.}},
\bauthor{\bsnm{Schr{\"{o}}der}, \binits{D.}}:
\batitle{{Eigenstate Thermalization Hypothesis for Wigner Matrices}}.
\bjtitle{Communications in Mathematical Physics}
\bvolume{388}(\bissue{2}),
\bfpage{1005}--\blpage{1048}
(\byear{2021})
\doiurl{10.1007/s00220-021-04239-z}
{\href{https://arxiv.org/abs/2012.13215}{{2012.13215}}}
\end{barticle}
\endbibitem

%%% 29
\bibitem[\protect\citeauthoryear{Cipolloni
  et~al.}{2023}]{cipolloni2023gaussian}
\begin{botherref}
\oauthor{\bsnm{Cipolloni}, \binits{G.}},
\oauthor{\bsnm{Erdős}, \binits{L.}},
\oauthor{\bsnm{Henheik}, \binits{J.}},
\oauthor{\bsnm{Kolupaiev}, \binits{O.}}:
{Gaussian fluctuations in the Equipartition Principle for Wigner matrices},
1
(2023)
{\href{https://arxiv.org/abs/2301.05181}{{2301.05181}}}
\end{botherref}
\endbibitem

%%% 30
\bibitem[\protect\citeauthoryear{Adhikari
  et~al.}{2023}]{adhikari2023eigenstate}
\begin{botherref}
\oauthor{\bsnm{Adhikari}, \binits{A.}},
\oauthor{\bsnm{Dubova}, \binits{S.}},
\oauthor{\bsnm{Xu}, \binits{C.}},
\oauthor{\bsnm{Yin}, \binits{J.}}:
{Eigenstate Thermalization Hypothesis for Generalized Wigner Matrices}
(2023)
{\href{https://arxiv.org/abs/2302.00157}{{2302.00157}}}
\end{botherref}
\endbibitem

%%% 31
\bibitem[\protect\citeauthoryear{Sugimoto et~al.}{2021}]{sugimoto2021test}
\begin{barticle}
\bauthor{\bsnm{Sugimoto}, \binits{S.}},
\bauthor{\bsnm{Hamazaki}, \binits{R.}},
\bauthor{\bsnm{Ueda}, \binits{M.}}:
\batitle{{Test of the Eigenstate Thermalization Hypothesis Based on Local
  Random Matrix Theory}}.
\bjtitle{Physical Review Letters}
\bvolume{126}(\bissue{12}),
\bfpage{120602}
(\byear{2021})
\doiurl{10.1103/PhysRevLett.126.120602}
\end{barticle}
\endbibitem

%%% 32
\bibitem[\protect\citeauthoryear{Erdős et~al.}{2013}]{erdos2013local}
\begin{barticle}
\bauthor{\bsnm{Erdős}, \binits{L.}},
\bauthor{\bsnm{Knowles}, \binits{A.}},
\bauthor{\bsnm{Yau}, \binits{H.-T.}},
\bauthor{\bsnm{Yin}, \binits{J.}}:
\batitle{{The local semicircle law for a general class of random matrices}}.
\bjtitle{Electronic Journal of Probability}
\bvolume{18}(\bissue{59}),
\bfpage{1}
(\byear{2013})
\doiurl{10.1214/EJP.v18-2473}
{\href{https://arxiv.org/abs/1212.0164}{{1212.0164}}}
\end{barticle}
\endbibitem

%%% 33
\bibitem[\protect\citeauthoryear{Collins et~al.}{2022}]{collins2022weingarten}
\begin{barticle}
\bauthor{\bsnm{Collins}, \binits{B.}},
\bauthor{\bsnm{Matsumoto}, \binits{S.}},
\bauthor{\bsnm{Novak}, \binits{J.}}:
\batitle{{The Weingarten Calculus}}.
\bjtitle{Notices of the American Mathematical Society}
\bvolume{69}(\bissue{05}),
\bfpage{1}
(\byear{2022})
\doiurl{10.1090/noti2474}
{\href{https://arxiv.org/abs/2109.14890}{{2109.14890}}}
\end{barticle}
\endbibitem

%%% 34
\bibitem[\protect\citeauthoryear{Sugimoto
  et~al.}{2022}]{sugimoto2022eigenstate}
\begin{barticle}
\bauthor{\bsnm{Sugimoto}, \binits{S.}},
\bauthor{\bsnm{Hamazaki}, \binits{R.}},
\bauthor{\bsnm{Ueda}, \binits{M.}}:
\batitle{{Eigenstate Thermalization in Long-Range Interacting Systems}}.
\bjtitle{Physical Review Letters}
\bvolume{129}(\bissue{3}),
\bfpage{030602}
(\byear{2022})
\doiurl{10.1103/PhysRevLett.129.030602}
\end{barticle}
\endbibitem

%%% 35
\bibitem[\protect\citeauthoryear{Hamazaki and
  Ueda}{2018}]{hamazaki2018atypicality}
\begin{barticle}
\bauthor{\bsnm{Hamazaki}, \binits{R.}},
\bauthor{\bsnm{Ueda}, \binits{M.}}:
\batitle{{Atypicality of Most Few-Body Observables}}.
\bjtitle{Physical Review Letters}
\bvolume{120}(\bissue{8}),
\bfpage{080603}
(\byear{2018})
\doiurl{10.1103/PhysRevLett.120.080603}
\end{barticle}
\endbibitem

\end{thebibliography}

\end{document}